\documentclass[11pt]{article}
\usepackage[margin=0.9in]{geometry}
\usepackage{bm}
\usepackage{graphicx}
\usepackage{enumerate}
\usepackage{amssymb,amsmath,amsthm,mathrsfs}
\usepackage{xcolor}
\usepackage{subcaption}
\usepackage{caption}
\usepackage{amsfonts,delarray}
\usepackage{rotating,color}
\usepackage{natbib,multirow}
\usepackage{titlesec,textcase}
\usepackage{longtable}
\usepackage{bbm}
\usepackage{rotating}
\usepackage{authblk}
\usepackage{times,url}
\usepackage{authblk}
\usepackage{algorithm2e}
\usepackage{booktabs}   
\usepackage[colorlinks=true, allcolors=blue]{hyperref}
\usepackage{cleveref}
\usepackage{float} 
\usepackage{placeins} 
\usepackage{makecell}

\newcommand{\beforecline}{\\[2pt]}        
\newcommand{\aftercline}{\rule{0pt}{10pt}} 

\graphicspath{{../fig/}}

\newcommand{\indep}{\perp \!\!\! \perp}

\newtheorem{lemma}{Lemma}
\newtheorem{condition}{Condition}
\newtheorem{definition}{Definition}
\newtheorem{proposition}{Proposition}
\newtheorem{theorem}{Theorem}

\newcommand{\bfm}[1]{\ensuremath{\mathbf{#1}}}

\def\ba{\bfm a}     
\def\bb{\bfm b}   \def\bB{\bfm B}

\def\bh{\bfm h}   \def\bH{\bfm H}

   \def\bP{\bfm P}

\def\bx{\bfm x}     
\def\by{\bfm y}

\newcommand{\bfsym}[1]{\ensuremath{\boldsymbol{#1}}}

\def\bbeta{\bfsym \beta}

\def\btheta{\bfsym \theta}

\renewcommand{\hat}{\widehat}
\renewcommand{\tilde}{\widetilde}

\DeclareMathOperator*{\argmin}{arg\,min}

\definecolor{rp}{RGB}{83,54,106}

\def\boxit#1{\vbox{\hrule\hbox{\vrule\kern6pt\vbox{\kern6pt#1\kern6pt}\kern6pt\vrule}\hrule}}

\begin{document}

\title{StaRQR-K: False Discovery Rate Controlled Regional Quantile Regression}
\author[1,2]{Sang Kyu Lee}
\author[1]{Tongwu Zhang}
\author[,1]{Hyokyoung G. Hong$^{*}$}
\author[,3]{Haolei Weng\footnote{denotes equal contribution}}
\affil[1]{Biostatistics Branch, Division of Cancer Epidemiology and Genetics, National Cancer Institute, National Institutes of Health}
\affil[2]{Department of Applied Statistics, Konkuk University}
\affil[3]{Department of Statistics and Probability, Michigan State University}

\date{}
\maketitle
\begin{abstract}
{Quantifying how genomic features influence different parts of an outcome distribution requires statistical tools that go beyond mean regression, especially in ultrahigh-dimensional settings. Motivated by the study of LINE-1 activity in cancer, we propose StaRQR-K, a stabilized regional quantile regression framework with model-X knockoffs for false discovery rate control. StaRQR-K identifies CpG sites whose methylation levels are associated with specific quantile regions of an outcome, allowing detection of heterogeneous and tail-sensitive effects. The method combines an efficient regional quantile sure independence screening procedure with a winsorizing-based model-X knockoff filter, providing false discovery rate (FDR) control for regional quantile regression. Simulation studies show that StaRQR-K achieves valid FDR control and substantially higher power than existing approaches. In an application to The Cancer Genome Atlas head and neck cancer cohort, StaRQR-K reveals quantile-region-specific associations between CpG methylation and LINE-1 activity that improve out-of-sample prediction and highlight genomic regions with known functional relevance.}
\end{abstract}


\sloppy

\section{Introduction}
\label{sec:intro}

{

Long Interspersed Nuclear Element 1 (LINE-1) retrotransposons are repetitive DNA sequences that can copy and insert themselves into new genomic locations. These insertions compromise genomic integrity and are frequently observed in cancer. Methylation at LINE-1 promoter regions typically suppresses such activity, whereas hypomethylation in tumors can reactivate LINE-1, leading to increased insertions. Because LINE-1 activity, reflected by the total number of transposable element (TE) insertions, indicates genomic instability, it has attracted growing interest as a marker of cancer development and progression. Understanding the molecular mechanisms that influence LINE-1 activation may offer new perspectives on the biological processes underlying tumorigenesis. This study investigates whether methylation levels at CpG sites within LINE-1 promoter regions are associated with elevated LINE-1 activity. Since LINE-1 activity cannot be directly observed, we use the total number of TE insertions as a proxy reflecting cumulative mobilization events. Of particular interest is whether specific CpG sites are linked to the upper tail of the TE insertion distribution, where genomic instability is most pronounced.

The statistical challenges are considerable. All samples in this study are derived from tumor tissues, where LINE-1 activity is already present to some degree. In our data, the total number of TE insertions is at least one for every sample, with the majority showing low counts and a few exhibiting extremely high counts. This distribution reflects substantial heterogeneity in LINE-1 activation among tumors and motivates modeling approaches that focus on the upper tail rather than on the mean. The analysis further involves approximately 11,500 CpG sites within or near LINE-1 promoter regions, each represented by its methylation level, while only a small subset is expected to be associated with elevated LINE-1 activity. This ultrahigh-dimensional structure requires a variable selection procedure capable of reliably identifying CpG sites linked to higher, rather than average, activity levels. To address these challenges, we propose a new framework based on regional quantile regression (RQR), which models associations across a range of quantiles rather than at a single level, thereby improving stability and revealing quantile-dependent signals. The method integrates two complementary components: a screening step with theoretical sure-screening guarantees that enhances computational scalability while retaining potentially relevant CpG sites, and a knockoff-based variable selection procedure that enables valid false discovery rate (FDR) control in ultrahigh-dimensional settings. This unified approach offers a principled and reproducible framework for identifying CpG sites associated with elevated LINE-1 activity.

Variable screening for ultrahigh-dimensional data has been widely studied since the seminal work on Sure Independence Screening (SIS), which ranks covariates by their marginal associations with the response \citep{fan2008sis}. Although SIS and its extensions \citep{fan2010glm, hall2009using, fan2011nonparametric} have been successful in various models, mean-based screening criteria are often sensitive to heavy-tailed errors and heteroscedasticity. Quantile regression (QR) provides a more robust alternative by characterizing associations across the conditional response distribution \citep{koenker1978regression}, motivating high-dimensional QR methods using penalization or marginal screening \citep{wang2012qr, he2013quantile, fan2014arlasso, ma2017qpcor, Lee2023Quantile, Park2023Varying}. However, most existing QR-based screening methods target a single quantile level, which can produce unstable selections when effects vary across quantiles \citep{zheng2015globally}. Regional quantile regression extends QR by modeling effects over an interval of quantiles, offering a more coherent view of distributional patterns \citep{PARK201713, zheng2015globally, yoshida2021}. Yet, no existing regional quantile screening methods incorporate formal FDR control. As noted by \cite{liu2020}, many high-dimensional screening procedures use conservative thresholds that compromise power or lack theoretical FDR guarantees. Building on the model-X knockoff framework \citep{barber2015fdr, candes2018modelx}, our method extends valid FDR control to the regional quantile regression setting, where standard knockoff constructions fail because the conditional independence assumption does not hold across quantiles. To address this, we introduce a restricted quantile representation, using winsorization, that confines the analysis to a subset of quantiles where the relationship between predictors and the outcome remains valid, allowing valid knockoff generation within this range. A subsequent stabilization step further enhances reproducibility and power. The resulting procedure provides robust and interpretable variable selection with FDR control in ultrahigh-dimensional regional quantile analysis. 

{This paper contributes in several ways. Most importantly, we develop the first statistical framework for variable selection in regional quantile regression that effectively targets FDR control in high-dimensional settings. Building on this main contribution, we add several specific components. First, we incorporate a regional quantile sure independence screening step that prioritizes features showing evidence of quantile-region–specific associations. Second, we establish a winsorizing-based transformation that enables the use of valid knockoff statistics in the knockoff filter. Third, we demonstrate the practical value of the method through an application to Head and Neck Squamous Cell Carcinoma (HNSC) data from The Cancer Genome Atlas (TCGA). LINE-1 activity is elevated in HNSC tumors and is widely recognized as a marker of genomic instability. Our analysis identifies CpG sites whose methylation relates to variation in LINE-1 activity, offering new statistical and biological insight into its epigenetic regulation in cancer.}

The remainder of the paper is organized as follows. Section \ref{sec:main} introduces the proposed screening procedure and the associated false discovery rate (FDR) control method for the regional quantile regression model, along with its computational and theoretical developments. Section \ref{c3:sec:sim} presents simulation studies showing that the proposed screening step yields smaller minimum model sizes than competing approaches and that the winsorizing-based knockoff procedure achieves valid FDR control with higher power and improved stability relative to naive knockoff baselines. Section \ref{real:data:sec} reports the results of the real data analysis in cancer genomics. Section \ref{sec:conclusion} concludes with a summary of the main findings. Section \ref{sec:append} provides some supplementary information, including proofs and a table. 

For clarity, we summarize the notation used throughout the paper. For $\ba = (a_1, \ldots, a_p)^\top \in \mathbb{R}^p$, let $|\ba|_q = \left(\sum_{i=1}^p |a_i|^q\right)^{1/q}$ for $q \in [1, \infty)$. Denote by $\lambda_{\max}(\cdot)$ and $\lambda_{\min}(\cdot)$ the largest and smallest eigenvalues of a matrix, respectively. Boldface letters indicate vectors. For a set $A$, $\mathbbm{1}_A(\cdot)$ denotes the indicator function, and $|A|$ its cardinality. For constants $a$ and $b$, $a\vee b$ denotes $\max\{a,b\}$ and $a\wedge b$ denotes $\min\{a,b\}$.}

\section{Feature Screening and FDR Control for Regional Quantile Regression} \label{sec:main}

We consider the following regional quantile regression form:
\begin{align}
    Q_{y|\bx}(\tau) = \alpha^*(\tau) + \bx^T\bbeta^*(\tau), \quad \tau \in \Delta, \label{eq:quantile model 1}
\end{align}
where $Q_{y|\bx}(\tau)=\inf \{t:\mathbb{P}(y\leq t | \bx)\geq \tau\}$ denotes the $\tau$th conditional quantile of a response variable $y\in \mathbb{R}$ given the covariates $\bx=(x_1,\ldots, x_p)\in \mathbb{R}^p$, and $\Delta = [\Delta_l, \Delta_u] \subseteq (0,1)$ is an interval of quantile levels of interest. The parameters $\alpha^*(\tau)$ and $\bbeta^*(\tau)=(\beta_1^*(\tau),\ldots,\beta_p^*(\tau))$ denote the intercept and coefficient functions, respectively, which are allowed to vary with the quantile level $\tau\in \Delta$. The set of active covariates is defined as
\begin{align}
\label{true:active}
    \mathcal{M}(\Delta) = \left\{ 1\leq j \leq p: \beta^*_j (\tau) \neq 0 \text{ for some } \tau \in \Delta \right\}.
\end{align}
According to the definition, an active variable may influence all or some quantiles in the region $\Delta$ of interest. Given $n$ independent and identically distributed observations, denoted by $\{(\bx_i,y_i)\}_{i=1}^n$, our goal is to identify the set of active covariates $\mathcal{M}(\Delta)$, i.e., the variables which impact the prespecified part (captured by $\Delta$) of the conditional distribution of the response. We adopt a common global sparsity assumption that the number of active variables is smaller than the sample size \citep{zheng2015globally, yoshida2021}, that is, $|\mathcal{M}(\Delta)|\leq n$.

We are concerned with the ultrahigh-dimensional setting where the covariate dimension $p$ can be exponentially large in relative to the sample size $n$. This is a typical scenario in genomics studies, where the number of assayed molecular features (e.g., SNPs, CpG sites, gene transcripts, metabolites) routinely reaches tens of thousands to millions while available samples are often in the hundreds. We develop a screen and clean procedure, where the first stage performs feature screening to reduce dimensionality from high to moderate scale and the second stage employs a knockoff filter to control false discovery rate for variable selection. We discuss in detail the two stages in Sections \ref{c3:subsec:2.1} and \ref{c3:subsec:2.4}, respectively. 

\subsection{Feature Screening via Integrated Regional Quantile Estimation}\label{c3:subsec:2.1}

This section introduces our proposed screening method for ultrahigh-dimensional data. Without loss of generality, we assume each covariate is standardized such that $\mathbb{E}x_{ij}=0, {\rm Var}(x_{ij})=1, 1\leq j \leq p$, where $x_{ij}$ is the $i$th observation for the $j$th covariate. Our method ranks the importance of each covariate by an integrated marginal quantile statistics. Specifically, let $\Delta_l=\xi_0<\xi_1<\cdots<\xi_{L_0-1}<\xi_{L_0}=\Delta_u$ be a partion of the interval $\Delta$, and $\mathcal{S}(k,\bfm{\xi})$ be the space of polynomial splines of order $k+1$ with knots $\{\xi_i\}_{i=1}^{L_0-1}$. There exists a normalized B-spline basis for $\mathcal{S}(k,\bfm{\xi})$ \citep{schumaker2007spline}, denoted by $\bB(\tau)=(B_1(\tau),\ldots,B_{N}(\tau))$ with $N=k+L_0$. For each $j=1,\ldots, p$, solve the following marginal quantile regression problem
\begin{align*}
    (\hat{{\ba}}_j,\hat{{\bb}}_j)\in \argmin_{{\ba},{\bb}}\sum_{\ell=1}^L\sum_{i=1}^n\rho_{\tau_{\ell}}\Big(y_i-x_{ij}\mathbf{B}(\tau_{\ell})^T{\bb}-\mathbf{B}(\tau_{\ell})^T{\ba}\Big), 
\end{align*}
where $\rho_{\tau}(u) = u(\tau-I(u\leq0))$ is the check loss function, and $\Delta_l<\tau_1<\cdots<\tau_{L}<\Delta_u$ are uniformly spaced quantile levels over the interval $\Delta$. The importance of each covariate is then measured by the integrated squared coefficient estimator,
\[
\hat{R}_j=\frac{1}{L}\sum_{\ell=1}^L (\bB(\tau_{\ell})^T\hat{\bb}_j)^2, \quad j=1,\ldots, p.
\]
Our screening step retains the subset of variables
\begin{align*}
    \widehat{\mathcal{M}}(\Delta) = \Big\{1\leq j\leq p: \hat{R}_j \geq \nu_0\Big\},
\end{align*}
where $\nu_0 > 0$ is a pre-specified cutoff threshold. 

The method is based on the magnitude of the estimated marginal regression coefficients. Utilizing marginal quantile regression statistics for screening is not new \citep{he2013quantile, ma2016asymptotic, xu2017model}. Our method differs from the existing ones by the smoothing (via B-splines) and aggregating (across quantile levels) components. In the context of regional quantile regression, the marginal regression coefficients tend to change smoothly with the quantile level $\tau$. Hence, our spline-based smoothing approach can lead to better estimators than the naive approach that estimates coefficients independently. Moreover, for an active variable, its marginal regression coefficient might be small for some quantiles. However, as long as the aggregated magnitude across quantiles is large enough, our method is able to detect it. We formalize these arguments in a theorem that establishes the sure screening property \citep{fan2008sis} -- all the active variables survive after variable screening with probability tending to one. Towards that end, for each $j \in \mathcal{M}(\Delta)$, define the population quantities:
\begin{align}
({\ba}_j^*,{\bb}_j^*)  &\in \argmin_{{\ba},{\bb}\in \mathbb{R}^N}\frac{1}{L}\sum_{\ell=1}^L\mathbb{E}\rho_{\tau_{\ell}}\Big(y-x_j\mathbf{B}(\tau_{\ell})^T{\bb}-\mathbf{B}(\tau_{\ell})^T{\ba}\Big), \label{star:def}\\
 (f_j(\tau),g_j(\tau)) &\in \argmin_{f,g\in \mathbb{R}}\mathbb{E}\rho_{\tau}(y-x_jf-g). \label{c3:pop:gf}
\end{align}

\noindent We impose the following regularity conditions.

\begin{condition}[\textbf{Bounded covariates}] \label{con:1}
The covariates are bounded: $\max_{j\in \mathcal{M}(\Delta)}|x_j|\leq K_x$.
\end{condition}

\begin{condition}[\textbf{Smoothness condition}] \label{con:2}
The functions $\{f_j(\tau),g_j(\tau)\}_{j\in \mathcal{M}(\Delta)}$ belong to the class of functions whose $k$th derivative satisfies a Lipschitz condition of order $c$: $|h^{(k)}(s)-h^{(k)}(t)|\leq c_0|s-t|^c$, where $k$ is a nonnegative integer and $c \in (0,1]$ satisfies $d=k+c>0.5$. 
\end{condition}

\begin{condition}[\textbf{Bounded density and its derivative}] \label{con:3}
For each $j\in \mathcal{M}(\Delta)$, the conditional density evaluated at $t=x_jf_j(\tau)+g_j(\tau)$ and $x_j \mathbf{B}(\tau)^T\bb_j^*+\mathbf{B}(\tau)^T\ba_j^*$ is bounded uniformly in $(x_j,\tau)$: $0<\underline{f}\leq p_{y|x_j}(t)\leq \bar{f}<\infty$. Moreover, its derivative is uniformly bounded: $\sup_t|p'_{y|x_j}(t)|\leq \bar{f}'$. 
\end{condition}

\begin{condition}[\textbf{Knot allocation}] \label{con:4}
$\frac{\max_{1\leq i\leq L_0}(\xi_i-\xi_{i-1})}{\min_{1\leq i\leq L_0}(\xi_i-\xi_{i-1})} \leq c_1$ for some constant $c_1\geq 1$. 
\end{condition}

\begin{condition}[\textbf{Strong integrated marginal signal}] \label{con:5}
$\min_{j\in \mathcal{M}(\Delta)}\frac{1}{L}\sum_{\ell=1}^L|f_j(\tau_{\ell})| \geq \kappa n^{-\gamma}$ for some constants $\kappa, \gamma>0$.
\end{condition}

\begin{condition}[\textbf{Scaling condition}] \label{con:6}
$Nn^{2\gamma-1}\log n=o(1), N^{-d}n^{\gamma}=o(1), Nn^{-1/2}\sqrt{\log n}=o(1), N=o(L)$. 
\end{condition}

Conditions \ref{con:1}-\ref{con:3} are standard in the literature on quantile regression and nonparametric spline approximation. \citep{belloni2011l1, he2013quantile, zheng2015globally, sherwood2016partially, yoshida2021}. Condition \ref{con:4} requires the spacing between knots to be of the same order, aiming to control the eigenvalues of the key B-spline related matrix $\sum_{\ell=1}^L\bB(\tau_{\ell})\bB(\tau_{\ell})^T$ \citep{shen1998local}. Condition \ref{con:5} is a common marginal signal condition in the quantile screening literature \citep{he2013quantile, ma2016asymptotic, xu2017model}. As noted earlier, our screening method integrates marginal signals across different quantiles, thus only requiring the integrated signal to be strong. Finally, Condition \ref{con:6} specifies the allowable growth rate of the number of B-spline basis functions and integrated quantiles relative to the sample size. For instance, choose $N=n^{c_N}, L=n^{c_L}$, then Condition \ref{con:6} is satisfied when $\gamma/d<c_N<\min(1-2\gamma, 1/2)$ and $c_L>c_N$. We also see that the smaller $\gamma$ (stronger marginal signal) and larger $d$ (smoother marginal coefficient) are, the weaker Condition \ref{con:6} is. Under these conditions, we establish the following sure screening property. 

\begin{theorem}
\label{c3:thm:one}
Under Conditions \ref{con:1}–\ref{con:6}, if we choose the threshold $\nu_0=\frac{\kappa^2}{5}n^{-2\gamma}$, then the sure screening property holds:
\[
\mathbb{P}\big(\mathcal{M}(\Delta)\subseteq \widehat{\mathcal{M}}(\Delta)\big)\rightarrow 1, \quad \text{ as  }~~  n\rightarrow \infty.
\]
\end{theorem}
The proof of Theorem \ref{c3:thm:one} is provided in Section \ref{c3:proof:thm1}. This theorem guarantees that, with a sufficiently large sample size, our procedure retains all active variables (i.e. no false negative error) with a probability approaching one. However, the selected set $\widehat{\mathcal{M}}(\Delta)$ may also include some inactive variables. The second stage is therefore to control the false positive error among the selected variables.

\subsection{FDR Control for the Regional Quantile Regression}\label{c3:subsec:2.4}

In the second stage, we focus on controlling the false discovery rate (FDR), defined as 
\begin{align}
\label{fdr:regional}
{\rm FDR}_{\Delta}=\mathbb{E}\bigg[\frac{|\{j: j\in \hat{S}_{\Delta} {\rm~and~} j\notin \mathcal{M}(\Delta)\}|}{|\hat{S}_{\Delta}|\vee 1}\bigg],
\end{align}
for a given variable selection outcome $\hat{S}_{\Delta}$ where $\hat{S}_{\Delta} \subset \{1,\dots,p\}$ is a subset of indexes for quantile region ${\Delta}$. We omit ${\Delta}$ if it is obvious afterward. FDR is the expected proportion of falsely selected variables. For instance, {with FDR controlled at $q = 0.1$, the expected proportion of false discoveries among the selected variables is at most 10 percent.}
The notion of FDR was proposed by \cite{benjamini1995controlling} and has been extensively studied in multiple testing \citep{benjamini2010discovering,efron2012large,efron2021computer}. For the problem of variable selection, earlier works can only guarantee FDR control under very strong assumptions \citep{miller2002subset,liu2010stability,meinshausen2010stability,bogdan2015slope,g2016sequential}. Until recently, a knockoff framework has been established to provide general FDR control guarantees for Gaussian linear models \citep{barber2015fdr} and then for arbitrary conditional distribution of the response \citep{candes2018modelx}. However, even the knockoff framework cannot be directly applied to our problem. In the rest of the section, we give some necessary background for the knockoff framework, explain the issue, and provide our solution.

\subsubsection{A Brief Review of the Model-X Knockoff Filter}\label{c3:subsec:2.2}

The model-X knockoff framework, introduced by \cite{candes2018modelx}, is general and flexible, as it works under high-dimensional settings and operates without making modeling assumptions about the conditional distribution of the response variable. In such (conditional) distribution-free settings, it is important to first clarify the meaning of active/inactive variables. 

\begin{definition}
\label{knock:active:def}
Given a response variable $y$ and covariates $\bx=(x_1,\ldots, x_p)$\footnote{Note that the number of variables left after the screening stage is much smaller than $p$. With a slight abuse of notation, we still use $p$ to denote the number of variables. Also, conditioning on the sure screening, the remaining variables still follow the regional quantile model.}, a variable $x_j$ is an inactive variable if and only if $y$ is independent of $x_j$ conditionally on the other variables $\{x_1,\ldots, x_p\}\setminus \{x_j\}$. The subset of inactive variables is denoted by $\mathcal{H}_0\subseteq \{1,2,\ldots, p\}$. A variable $x_j$ is active if and only if $j\notin \mathcal{H}_0$.
\end{definition}


The key idea of the model-X knockoff framework is to construct knockoff variables as controls to tease apart active variables from inactive ones. We say that a newly constructed random vector $\tilde{\bx}\in \mathbb{R}^p$ is the model-X knockoffs for the orignal variables $\bx \in \mathbb{R}^p$ if it has the two properties as follows:
\begin{enumerate}
    \item Pairwise exchangeability: $(\bx, \tilde{\bx}) \overset{d}{=} (\bx, \tilde{\bx})_{swap(S)}, \quad \forall S\subseteq \{1,2,\ldots, p\}$,
    \item Conditional independence: $\tilde{\bx} \indep \by|\bx$,
\end{enumerate}
where $swap(S)$ means swapping the entries $x_j$ and $\tilde{x}_j$ for each $j\in S$. Given i.i.d. observations $\{(\bx_i,y_i)\}_{i=1}^n$, a knockoff $\tilde{\bx}_i$ is constructed independently for $\bx_i$, $i=1,2,\ldots, n$. The augmented data $\{(\bx_i,\tilde{\bx}_i,y_i)\}_{i=1}^n$ is then used to compute importance statistics $W_j$ for each $j=1,\ldots, p$, and variables whose $W_j$ pass certain threshold are eventually selected. We refer to the original paper for a comprehensive treatment of the model-X knockoff framework. We will provide more details when it comes to our proposed method.

\subsubsection{Model-X Knockoff Filter for the Regional Quantile Regression}\label{c3:subsec:2.3}

Since the model-X knockoff framework does not require any assumption on the conditional distribution of the response, it can be applied to the regional quantile regression model as well. However, a critical issue arises as the set of active variables $\mathcal{M}(\Delta)$ (see \eqref{true:active}) in the context of regional quantile regression is not necessarily the same as the one (see Definition \ref{knock:active:def}) in the knockoff framework, as shown in the following proposition.

\begin{proposition}
\label{key:issue:raise}
Under the regional quantile regression model \eqref{eq:quantile model 1}, if $x_j$ is independent of $y$ conditionally on $\{x_1,\ldots,x_p\}\setminus \{x_j\}$, then $\beta^*_j(\tau)=0, \forall \tau \in \Delta$; but not vice versa.

\end{proposition}

\begin{proof}
If $x_j$ is independent of $y$ conditionally on $\{x_1,\ldots,x_p\}\setminus \{x_j\}$, $x_j$ does not influence the conditional quantiles in $\Delta$, hence $\beta^*_j(\tau)=0, \forall \tau \in \Delta$. On the other hand, even if $\beta^*_j(\tau)=0, \forall \tau \in \Delta$, $x_j$ may still impact some conditional quantiles outside $\Delta$, so that $x_j$ is not conditionally independent of $y$.
\end{proof}

Note that the FDR in the knockoff framework is defined as \eqref{fdr:regional} with $\mathcal{M}(\Delta)$ replaced by $\mathcal{H}^c_0$. Proposition \ref{key:issue:raise} shows $\mathcal{H}_0 \subseteq \mathcal{M}(\Delta)^c $, implying that the FDR the model-X knockoff can control is in fact smaller than the FDR \eqref{fdr:regional} that we aim to control. As a result, if we were to directly apply the model-X knockoff framework, we would not necessarily be able to control FDR at a target level. The underlying issue is that the coefficients $\bbeta^*(\tau)$ do not encode any information about the segment of the conditional distribution outside $\Delta$. We employ a winsorizing to address the issue, as demonstrated by the proposition below.

\begin{proposition} \label{c3:prop:censor}
Let $\tilde{y} \in \mathbb{R}$ be constructed as follows:
\begin{align*}
    \tilde{y} = 
    \begin{cases}
        Q_{y|\textbf{x}} (\Delta_l) & \text{ if } y < Q_{y|\textbf{x}} (\Delta_l) \\
        y & \text{ if }  Q_{y|\textbf{x}} (\Delta_l) \leq y \leq Q_{y|\textbf{x}} (\Delta_u) \\
        Q_{y|\textbf{x}} (\Delta_u) & \text{ if } y > Q_{y|\textbf{x}} (\Delta_u)
    \end{cases}
\end{align*}
Under the regional quantile regression model \eqref{eq:quantile model 1}, it holds that
\begin{align*}
Q_{\tilde{y}|\bx}(\tau)=
\begin{cases}
        \alpha^*(\Delta_l)+\bx^T\bbeta^*(\Delta_l) & \text{ if } \tau \leq \Delta_{l} \\
        \alpha^*(\tau)+\bx^T\bbeta^*(\tau) & \text{ if } \Delta_l< \tau < \Delta_u \\
        \alpha^*(\Delta_u)+\bx^T\bbeta^*(\Delta_u) & \text{ if } \tau \geq \Delta_u
    \end{cases}
\end{align*}
As a result, the variable $x_j$ is conditionally independent of $\tilde{y}$ if and only if $\beta^*_j(\tau)=0, \forall\tau\in \Delta = [\Delta_l, \Delta_u]$. 
\end{proposition}


The proof of Proposition \ref{c3:prop:censor} can be found in the supplementary material (Section \ref{sec:append}). Proposition \ref{c3:prop:censor} provides a recipe to adapt the knockoff framework for regional quantile regression. Instead of using the original response variable $y$, we use the winsorized version $\tilde{y}$ together with the covariates $\bx$ to run the model-X knockoff procedure. As shown in Section \ref{c3:sec:sim}, the winsorizing operation not only helps to achieve desired FDR control, but also enhances power.

To make the idea practical, we will use part of the data to estimate the unknown conditional quantiles that determine the winsorizing cut-offs in Proposition \ref{c3:prop:censor}. Standard linear quantile regression estimators often exhibit the well-known ``quantile crossing" problem, and if crossing occurs, the winsorizing operation in Proposition \ref{c3:prop:censor} may not be well defined. To avoid this issue, we adopt a non-crossing strategy similar to the ones in \citet{bondell2010,liu2011simultaneous}, to compute the non-crossing quantile LASSO, 
\begin{align}
\label{non:cross:lasso}
\{\hat{\alpha}(\tau_l),\hat{\bbeta}(\tau_l)\} = \argmin_{\{\alpha(\tau_l),\bbeta(\tau_l)\}} \sum_{l=1}^L \sum_{i\in \mathcal{I}_1} \rho_{\tau_l}(y_i - \alpha(\tau_{l})-\bx_i^T\bbeta(\tau_l)) + \lambda\sum_{l=1}^L\|\bbeta(\tau_l)\|_1 \\
\text{ subject to } \alpha(\tau_{l})+\bx^T\bbeta(\tau_l) \geq \alpha(\tau_{l-1})+ \bx^T\bbeta(\tau_{l-1}) \text{ for } \bx \in D,~l=2,\dots,L, \nonumber
\end{align}
where $D \subseteq \mathbb{R}^{p}$ is a closed convex polytope, $\mathcal{I}_1\subseteq \{1,2,\ldots,n\}$, and $\Delta_l=\tau_1<\tau_2<\cdots <\tau_{L-1}<\tau_L =\Delta_u$.
This non-crossing quantile LASSO guarantees monotone conditional quantile curves across the multiple quantile levels $\{\tau_l\}_{l=1}^L$, ensuring $\widehat{Q}_{y|\textbf{x}}(\Delta_l)=\hat{\alpha}(\Delta_l)+\bx^T\hat{\bbeta}(\Delta_l)  \leq \widehat{Q}_{y|\textbf{x}} (\Delta_u)=\hat{\alpha}(\Delta_u)+\bx^T\hat{\bbeta}(\Delta_u)$.

After obtaining the estimated conditional quantiles, we use the other part of the data, indexed by $\mathcal{I}_2$, to calculate winsorized responses, generate knockoff variables and compute important statistics. Specifically, the winsorized response variables are: $\forall i\in \mathcal{I}_2$,
\begin{align}
\label{censor:response:form}
    \tilde{y}_i = 
    \begin{cases}
        \hat{Q}_{y|\textbf{x}_i} (\Delta_l) & \text{ if } y_i < \hat{Q}_{y|\textbf{x}_i} (\Delta_l) \\
        y_i & \text{ if }  \hat{Q}_{y|\textbf{x}_i} (\Delta_l) \leq y_i \leq \hat{Q}_{y|\textbf{x}_i} (\Delta_u) \\
        \hat{Q}_{y|\textbf{x}_i} (\Delta_u) & \text{ if } y_i > \hat{Q}_{y|\textbf{x}_i} (\Delta_u)
    \end{cases}
\end{align}

We follow \cite{candes2018modelx} to generate second-order model-X knockoﬀs, denoted by $\{\tilde{\bx}_i\}_{i\in \mathcal{I}_2}$. To quantify the evidence for each variable's importance, we compute the non-crossing quantile LASSO as in \eqref{non:cross:lasso}, but based on the augmented data $\{(\bx_i,\tilde{\bx}_i,\tilde{y}_i)\}_{i\in \mathcal{I}_2}$: 
\begin{align}
\label{importance:lasso:def}
&\{\hat{\alpha}(\tau_l),\hat{\bbeta}(\tau_l),\hat{\tilde{\bbeta}}(\tau_l)\}  \\
= &\argmin_{\{\alpha(\tau_l),\bbeta(\tau_l),\tilde{\bbeta}(\tau_l)\}} \sum_{l=1}^L \sum_{i\in \mathcal{I}_2} \rho_{\tau_l}(\tilde{y}_i - \alpha(\tau_{l})-\bx_i^T\bbeta(\tau_l)-\tilde{\bx}_i^T\tilde{\bbeta}(\tau_l)) + \lambda\sum_{l=1}^L(\|\bbeta(\tau_l)\|_1+\|\tilde{\bbeta}(\tau_l)\|_1)\nonumber \\
&\text{ subject to } \alpha(\tau_{l})+\bx^T\bbeta(\tau_l)+\tilde{\bx}_i^T\tilde{\bbeta}(\tau_l) \geq \alpha(\tau_{l-1})+ \bx^T\bbeta(\tau_{l-1})+\tilde{\bx}_i^T\tilde{\bbeta}(\tau_{l-1})\text{ for } (\bx,\tilde{\bx}) \in D,~l=2,\dots,L. \nonumber
\end{align}
Then, the important statistics $W_j$ is defined as a weighted difference of the estimated coefficient magnitudes across the quantile region 
\begin{align}
\label{importance:def}
    W_j = \sum_{l=1}^{L}  w_l\left\{ \left|\widehat{\beta}_{j}(\tau_l)\right| - \left|\widehat{\tilde{\beta}}_{j},(\tau_l)\right| \right\}, \quad j=1,\ldots, p,
\end{align}
where $w_l$'s are non-negative weights. The intuition behind this statistic is straightforward. A large and positive value of $W_j$ indicates that the original variable $x_j$ exhibits a substantially stronger estimated effect across the quantile region than its knockoff $\tilde{x}_j$, providing evidence that $x_j$ is a genuine signal. We use equal weights in $W_j$, when signal information across quantile levels is unknown a priori.

Given the importance statistics $\{W_j\}_{j=1}^p$, we select variables in the set
\begin{align}
\label{knockoff:select:set}
\hat{\mathcal{M}} = \{ j : W_j \geq \omega \},
\end{align}
with a data-dependent threshold.
\begin{align}
\label{data:dep:thr}
    \omega = \min \left\{ t >0: \frac{\# \left\{  j : W_j \leq -t \right\} }{\# \left\{ j : W_j \geq t  \right\}} \leq q \right\},
\end{align}
where $q$ is the target FDR level. If the above set is empty, $\omega$ is set to $\infty$. While this threshold only leads to the provable control of a modified FDR, we still use it in our method, as it often shows effective control of the usual FDR empirically \citep{barber2015fdr, candes2018modelx}.


\subsection{The StaRQR-K Algorithm: A Step-by-Step Implementation}\label{sec:algorithm}

\RestyleAlgo{ruled}
\SetKwComment{Comment}{/* }{ */}
\SetKw{TuningParam}{Predefined Parameters:}
\SetKw{QuantRegion}{Quantile region of interest:}
\begin{algorithm}[t] 
\small
\caption{Stabilized Regional Quantile Regression with Knockoff Filter (StaRQR-K)}\label{c3:alg:RQR_knockoff_stable}
\KwData{$\{(\bx_i,y_i)\}_{i=1}^n$}
\QuantRegion{$\Delta = [\Delta_l, \Delta_u]$} 
\vspace{0.2cm}
\begin{enumerate}
    \item \parbox[t]{\dimexpr\linewidth-2.5em}{Randomly partition the dataset into three disjoint subsets: $\{1,2,\ldots, n\} =\cup_{i=0}^2 \mathcal{I}_i$, with corresponding sample sizes $\{n_i\}_{i=0}^2$.}
    \item \parbox[t]{\dimexpr\linewidth-2.5em}{Using the first subset $\mathcal{I}_0$, apply the screening step described in Section \ref{c3:subsec:2.1} to select the top $d_n$ variables. Only these selected variables will be used in the remaining steps.}
    \item[3.] {Repeat the following steps for $r=1,\dots,R$:}
    \begin{enumerate}
        \item \parbox[t]{\dimexpr\linewidth-2.5em}{Randomly partition the remaining data $\mathcal{I}_1 \cup \mathcal{I}_2$ into two new subsets, $\mathcal{I}^{(r)}_1$ and $\mathcal{I}^{(r)}_2$, with respective sample sizes $n_1$ and $n_2$.}
        \item \parbox[t]{\dimexpr\linewidth-2.5em}{On the subset $\mathcal{I}^{(r)}_1$, perform non-crossing quantile LASSO in \eqref{non:cross:lasso} to obtain $\widehat{Q}_{y|\textbf{x}}(\Delta_l)$ and $\widehat{Q}_{y|\textbf{x}} (\Delta_u)$.}
        \item \parbox[t]{\dimexpr\linewidth-2.5em}{Using the subset $\mathcal{I}^{(r)}_2$, construct the winsorizeed responses in \eqref{censor:response:form}, compute the importance statistics via \eqref{importance:lasso:def}-\eqref{importance:def}, and select variables by \eqref{knockoff:select:set}-\eqref{data:dep:thr}.}
        \item \parbox[t]{\dimexpr\linewidth-2.5em}{Store the resulting set of selected variables, $\hat{\mathcal{M}}^{(r)}_{FDR}$.}
    \end{enumerate}
    \item[4.] \parbox[t]{\dimexpr\linewidth-2.5em}{After $R$ iterations, calculate the selection probability for each variable:
    \begin{align*}
        \pi_j = \frac{1}{R}\sum_{r=1}^R \mathbbm{1}(j \in \hat{\mathcal{M}}^{(r)}_{FDR}),~j=1,\dots,p.
    \end{align*}}
    \item[5.] \parbox[t]{\dimexpr\linewidth-2.5em}{The final set of selected variables is given by $\hat{\mathcal{M}}_{FDR} = \{1\leq j \leq p: \pi_j \geq \eta\}$.}
\end{enumerate}
\end{algorithm}

A potential limitation of the knockoff filter we developed in Section \ref{c3:subsec:2.4}, however, stems from its reliance on a single generation of knockoff random variables and a single data splitting. This algorithmic randomness may lead to unstable results (since multiple runs of the method on the same dataset can yield different sets of selected variables), as well as a loss of statistical power (due to reduced sample size). Motivated by \cite{meinshausen2010stability} and \cite{ren2023derandomizing}, we aggregate the selection results over multiple repetitions, to improve stability and statistical power. We are now in the position to combine all the methodological components described so far into a coherent and complete procedure, termed Stabilized Regional Quantile Regression with Knockoff Filter (StaRQR-K). It is detailed in Algorithm \ref{c3:alg:RQR_knockoff_stable}.

\section{Simulation Studies}
\label{c3:sec:sim}

In this section, we present extensive simulations to empirically validate the performance of the proposed StaRQR-K procedure (see Algorithm \ref{c3:alg:RQR_knockoff_stable}). The performance of our methods will be benchmarked against several existing approaches to highlight the effectiveness of our methods. To be comprehensive, we evaluate the screening stage and FDR control stage of our procedure separately. After establishing the empirical validity of our approach in the simulations, we will demonstrate its practical utility on a challenging real-world dataset in Section \ref{real:data:sec}.

\begin{figure}[t]
    \centering
    \begin{subfigure}[t]{0.35\textwidth}
        \centering
        \includegraphics[width=\textwidth]{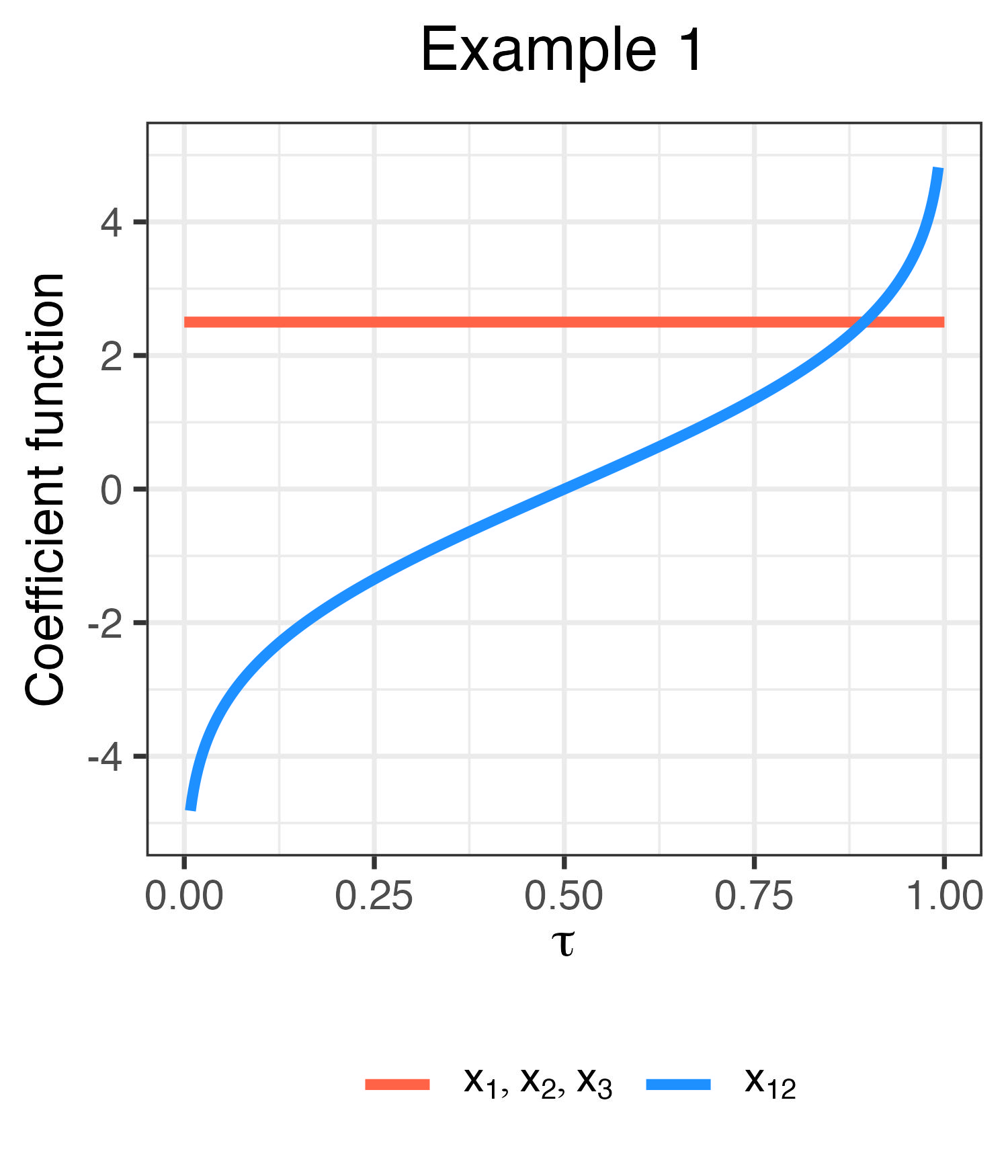}
    \end{subfigure}%
    ~ 
    \begin{subfigure}[t]{0.35\textwidth}
        \centering
        \includegraphics[width=\textwidth]{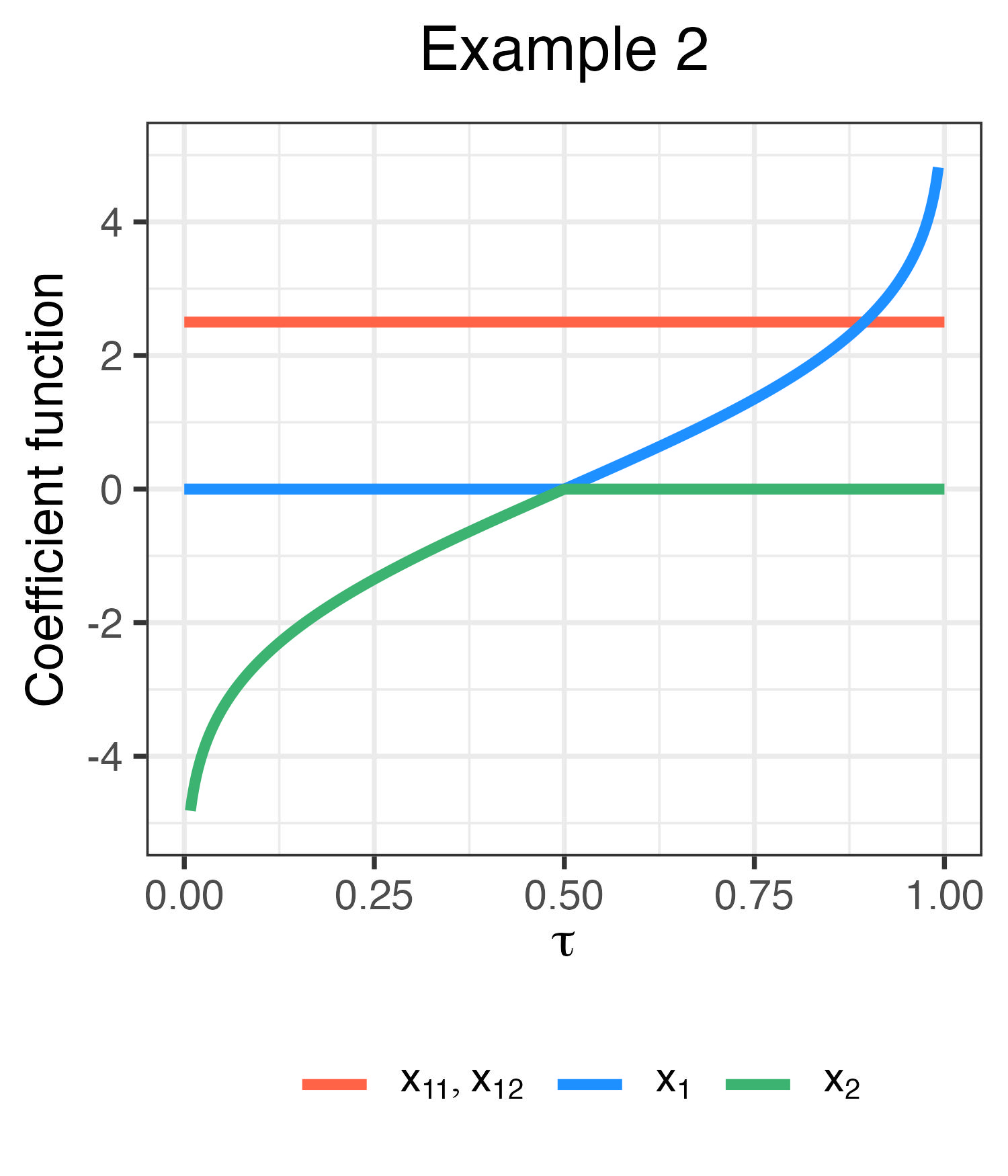}
        \label{fig2b:green}
    \end{subfigure}
    \caption{Coefficient functions for Examples 1 and 2. Red horizontal lines denote the homoskedastic coefficient functions, while blue and green curves denote the heteroskedastic coefficient functions.}
    \label{c3:fig2:hetero}
\end{figure}

\subsection{Screening Performance Comparison}

To evaluate the performance of our proposed feature screening method in Section \ref{c3:subsec:2.1}, we conduct simulation studies that compare it against several competing approaches across a variety of scenarios. We compare our method, regional quantile sure independence screening (rqSIS), with some popular model-free or quantile-model-based screening methods: DC-SIS \citep{li2012feature}, PC-Screen \citep{zhu2017projection}, qaSIS \citep{he2013quantile}, and QPCS \citep{ma2017qpcor}. For each model and setting, the procedure is repeated 500 times. Within each replication, we rank the features in descending order according to each screening method and record the minimum model size (MMS) required to contain all active features. To summarize these results, we report the 5\%, 35\%, 65\%, and 95\% quantiles of the MMS values across the 500 repetitions, from best to worst.

We fix $p=10,000, n=300$ and the quantile region of interest as $\Delta = [0.05, 0.95]$. For the methods, qaSIS and QPCS, since they are designed to deal with only one specific single quantile level, we set $\tau=0.25, 0.5, 0.75$ for them to compare the results. For the tuning parameters, we set $k=3, L_0=3, L=10$ in our method. For other methods, we use default or some tuned parameters based on the suggestions from the papers. 
We consider two examples:
\begin{itemize}
\item Example 1. The intercept function $\alpha^*(\tau)=0$. Only $x_1, x_2, x_3, x_{12}$ have nonzero coefficients, which are plotted on the left panel of Figure \ref{c3:fig2:hetero}. It features a simple heteroskedastic model where the coefficient $\beta_{12}(\tau)$ changes smoothly across the quantile levels $\tau$. 
\item Example 2. The intercept function $\alpha^*(\tau)=0$. Only $x_1, x_2, x_{11}, x_{12}$ have nonzero coefficients, which are plotted on the right panel of Figure \ref{c3:fig2:hetero}. It presents a more complex scenario where the coefficients $\beta_{1}(\tau)$ and $\beta_{2}(\tau)$ are nonzero only within part of the quantile region.
\end{itemize}
For both examples, we generate $\bx=(x_1,\ldots, x_p)$ in the following way: first sample $(z_1,\ldots, z_p)\sim \mathcal{N}(\textbf{0}, \Sigma)$, where each $(i,j)$th component of $\Sigma$ is defined by $\Sigma_{ij} = \sigma^{|i-j|}$ for $i,j=1,\dots,p$ and $\sigma = 0.5, 0.8$; then set $x_j=\Phi(z_j)$ for active variables and $x_j=z_j$ for inactive ones, where $\Phi(\cdot)$ is the standard normal distribution function. The implementation of our method in \verb|R| and \verb|Rcpp| is available at:
\href{https://github.com/SangkyuStat/StaRQRK}{https://github.com/SangkyuStat/StaRQRK}.




\begingroup
\renewcommand{\arraystretch}{1.5}
\begin{table}[!t] 
\centering 
\caption{Minimum model size (MMS) result with $p=10,000$ and $n=300$. The QPCS method can only track at most $n=300$ top variables, so if the MMS exceeded the sample size, the result is reported as `300+'.} 
\label{c3:tab:example2_p10000} 
\begin{tabular}{@{\extracolsep{0.1pt}} cc| cccc | cccc} 
\hline 
\hline 
& &\multicolumn{4}{c|}{$\sigma = 0.5$}  &\multicolumn{4}{c}{$\sigma = 0.8$}\\ 
& &5\% & 35\% & 65\% & 95\% &5\% & 35\% & 65\% & 95\% \\
\hline 
rqSIS & Example 1 & 4 & 4 & 4 & 4 & 4 & 4 & 5 & 7 \\ 
$(\Delta)$ & Example 2 & 4 & 4 & 4 & 4 & 4 & 4 & 4 & 6 \\ \hline
DC-SIS & Example 1 & 86.4 & 635 & 2105.4 & 6906 & 14 & 211.6 & 1106 & 6026.4 \\
 & Example 2 & 111.9 & 851.3 & 2196.6 & 6294.6 & 27.95 & 302.15 & 971.45 & 4257.55 \\  \hline
PC-Screen & Example 1 & 146 & 1126.8 & 2652.4 & 7353 & 15 & 309.4 & 1400.6 & 5812.8 \\ 
 & Example 2 & 294 & 1365.7 & 2958 & 6682.4 & 39.9 & 631.55 & 1714.1 & 5111.65 \\ \hline 
qaSIS & Example 1 & 88 & 1868.6 & 4643.2 & 9322.2 & 170.4 & 2520.6 & 5610.2 & 9224.4 \\  
 $(\tau=0.25)$ & Example 2 & 745.2 & 4146.1 & 6834.3 & 9616 & 903.3 & 5145.9 & 7613.85 & 9614.5 \\ \hline
qaSIS  & Example 1 & 438 & 2810.2 & 5753 & 9506 & 44 & 990 & 3749.2 & 8628.2 \\ 
 $(\tau=0.5)$ & Example 2 & 1161.6 & 4276.5 & 6704.8 & 9572.8 & 225.45 & 2631.95 & 5770.45 & 9314.55 \\ \hline
qaSIS  & Example 1 & 148.2 & 2019.4 & 4997.8 & 8908 & 22 & 438 & 2208.4 & 7571.8 \\ 
$(\tau=0.75)$ & Example 2 & 964.1 & 4258.3 & 6785.4 & 9478.8 & 51.95 & 958.25 & 2915.9 & 8285.65 \\ \hline
QPCS  & Example 1 & 4 & 4 & 4 & 300+ & 4 & 4 & 4 & 300+ \\ 
$(\tau=0.25)$ & Example 2 & 126.7 & 300+ & 300+ & 300+ & 272.8 & 300+ & 300+ & 300+ \\ \hline
QPCS  & Example 1 & 300+ & 300+ & 300+ & 300+ & 300+ & 300+ & 300+ & 300+ \\ 
$(\tau=0.5)$ & Example 2 & 300+ & 300+ & 300+ & 300+ & 300+ & 300+ & 300+ & 300+ \\ \hline
QPCS  & Example 1 & 4 & 4 & 4 & 300+ & 4 & 4 & 4 & 300+ \\ 
$(\tau=0.75)$ & Example 2 & 111.5 & 300+ & 300+ & 300+ & 272.85 & 300+ & 300+ & 300+ \\ 
\hline \hline \\[-1.8ex] 
\end{tabular} 
\end{table} 
\endgroup

The two examples are designed to challenge existing methods. In Example 1, it is expected that single-quantile-model methods (qaSIS and QPCS) will fail to detect $\beta_{12}(\tau)$ at $\tau = 0.5$, where its true value is zero. Similarly, in Example 2, the localized effects of $\beta_1(\tau)$ and $\beta_2(\tau)$ are expected to be difficult to identify. Also, the model-free methods (DC-SIS and PC-Screen) may not be as efficient as ours since our method exploits the model structure across different quantile levels. Indeed, this is confirmed in the simulation results from Table \ref{c3:tab:example2_p10000} (result for $p=5,000$ can be found in the supplementary material Table \ref{c3:tab:example2_p5000}), which demonstrate that our method is superior to the competing approaches in both examples.

\subsection{FDR Control Performance Comparison}

In Section \ref{c3:subsec:2.4}, we demonstrated why the standard knockoff procedure may not work for regional quantile regression, and developed a modified version via winsorizing. The main objective of this section is to evaluate the performance of the standard knockoff filter and our variant, in terms of FDR control and selection power, to validate the superiority of our proposed method. Throughout the section, we assume that we have already selected $p$ variables from the screening stage.


We set the target quantile region to $\Delta = [0.7, 0.9]$, and consider the following conditional quantile form:
\begin{align*}
    Q_{y|\bx}(\tau) = \sum_{j=1}^{s} ax_j\beta_j + 
    \sum_{j=p-K+1}^{p} bx_j\beta_j(\tau), 
\end{align*}
where $s$ is the number of homoskedastic coefficients, $K$ is the number of heteroskedastic coefficients, and $a$ is a coefficient inflation factor that will be used to adjust the signal-to-noise ratio (SNR). The covariates are generated as follows: firstly, sample $(z_1,\ldots, z_p)\sim \mathcal{N}(\textbf{0}, \Sigma)$, where each $(i,j)$th component of $\Sigma$ is defined by $\Sigma_{ij} = \sigma^{|i-j|}$ for $i,j=1,\dots,p$; then set $x_j=\Phi(z_j),~j=p-K+1,\ldots, p$ and $x_j=z_j,~j=1,\ldots, p-K$, where $\Phi(\cdot)$ is the standard normal distribution function. The homoskedastic coefficients $\beta_j,~j=1,\dots,s$ are set to 1. We consider two examples for the choice of heteroskedastic coefficients:

\begin{itemize}
\item Example 3. All the coefficients $\{\beta_j(\tau)\}_{j=p-K+1}^p$ are equal to the function colored in green from the right panel of Figure \ref{c3:fig2:hetero}. Since the quantile region of interest is $\Delta=[0.7,0.9]$, this setup makes the last $K$ variables inactive, although they are not conditionally independent of $y$. This is exactly the scenario discussed in Section \ref{c3:subsec:2.4}, where the standard (derandomized) knockoff filter may lose FDR control. 
\item Example 4. All the coefficients $\{\beta_j(\tau)\}_{j=p-K+1}^p$ are equal to the function colored in blue from the right panel of Figure \ref{c3:fig2:hetero}. In this setting, the last $K$ variables are active. Even so, we will see that our method has some improvement over the standard knockoff procedure. 
\end{itemize}

We vary $s, a, \sigma$ to evaluate the performances across scenarios with various sparsity, SNR and design correlation levels. Five methods are considered for comparison:
\begin{enumerate}
    \item NoWinsor: Run the standard knockoff procedure directly for the original response variables.
    \item NoWinsor-Stable: Repeat NoWinsor $R$ times and run Steps 4-5 from Algorithm \ref{c3:alg:RQR_knockoff_stable}, with $R=10, \eta=0.5$. 
    \item WinsorEst-Stable: Run Steps 3-5 from Algorithm \ref{c3:alg:RQR_knockoff_stable}, with $R = 10,\eta = 0.5$ and $n_1=n_2=n/2$.
    \item WinsorTrue: Winsor the response using the true conditional quantiles and then run the standard knockoff procedure for the winsorized data. 
    \item WinsorTrue-Stable: Repeat WinsorTrue $R$ times and run Steps 4-5 from Algorithm \ref{c3:alg:RQR_knockoff_stable}, with $R=10, \eta=0.5$. WinsorEst-Stable can be viewed as an approximation of WinsorTrue-Stable.     
\end{enumerate} 
Note that both WinsorTure and WinsorTrue-Stable use the true conditional quantile information, hence are not practical methods. We include them as a benchmark in the comparison. The target FDR level is set to $q=0.2$, and the tuning parameter for the non-crossing LASSO is selected using 5-fold cross-validation (CV). Each scenario is replicated 100 times, and we report the average false discovery proportion and true positive proportion as (approximate) FDR and power, respectively. For each example, we investigate three scenarios. For Example 3: (1) varying $s$ while fixing $a = 1/5$ and $\sigma = 0$; (2) varying $a$ while fixing $s=9$ and $\sigma = 0$; and (3) varying $\sigma$ while fixing $a = 1/5$ and $s=9$. We fix $b=1$ for all three scenarios in Example 3 to make it hard to find $s$ true constant signals. For Example 4: (1) varying $s$ while fixing $a = 1/5$ and $\sigma = 0$; (2) varying $a$ while fixing $s=9$ and $\sigma = 0$; and (3) varying $\sigma$ while fixing $a = 1/5$ and $s=9$. We fix $b=0.2$ for all three scenarios in Example 4 to make it hard to find $K$ true heteroskedastic signals. The results for Examples 3 and 4 are presented in Figures \ref{fig3:ex3.1} and \ref{fig4:ex3.2}, respectively.

\begin{figure}[t!]
    \centering
    \begin{subfigure}[t]{\textwidth}
        \centering
        \includegraphics[width=\textwidth, height=0.24\textheight]{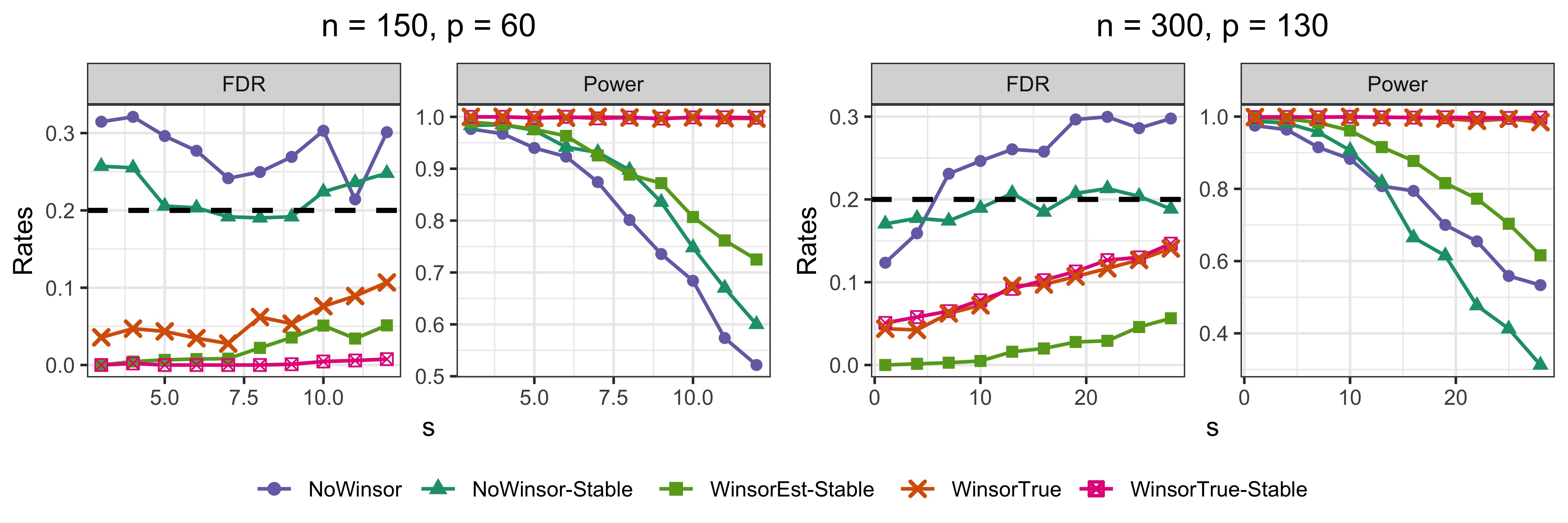}
    \end{subfigure}%
    \\
    \begin{subfigure}[t]{\textwidth}
        \centering
        \includegraphics[width=\textwidth, height=0.24\textheight]{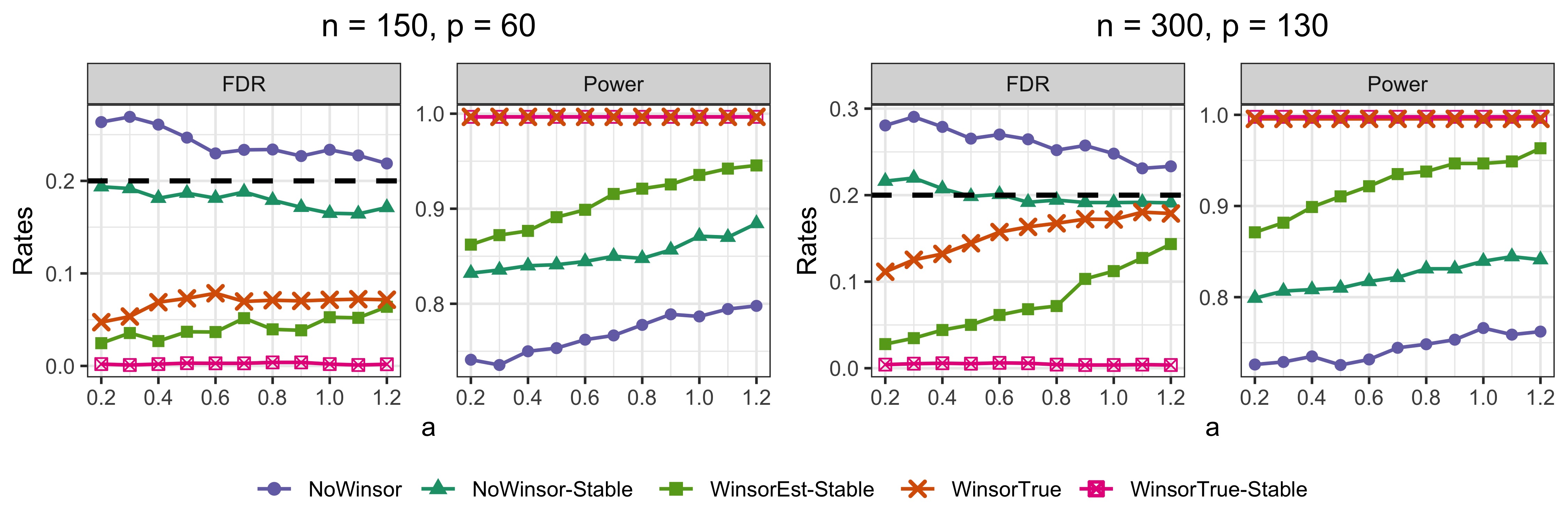}
    \end{subfigure}%
    \\
    \begin{subfigure}[t]{\textwidth}
        \centering
        \includegraphics[width=\textwidth, height=0.24\textheight]{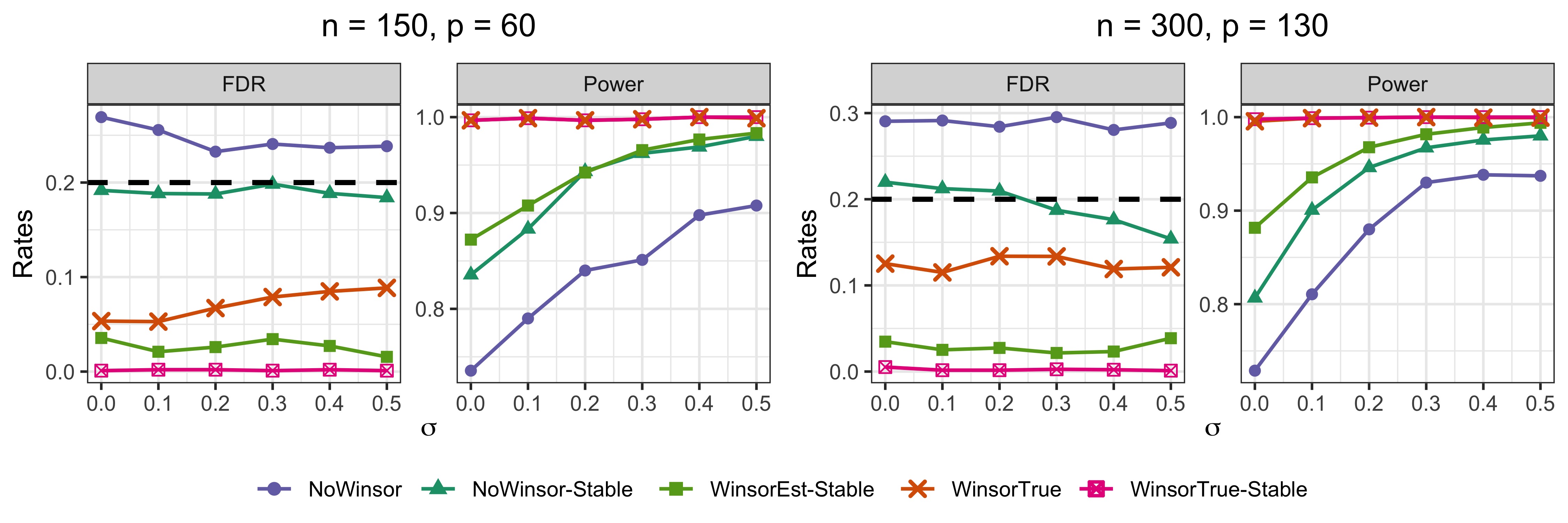}
    \end{subfigure}%
    \caption{Results obtained by varying $s$, $a$, $\sigma$ for $n \in \{150, 300\}$, and $p \in \{60, 130\}$. For all simulations, we fix $K = 0.6p$. }
    
    \label{fig3:ex3.1}
\end{figure}

\begin{figure}[t!]
    \centering
    \begin{subfigure}[t]{\textwidth}
        \centering
        \includegraphics[width=\textwidth, height=0.24\textheight]{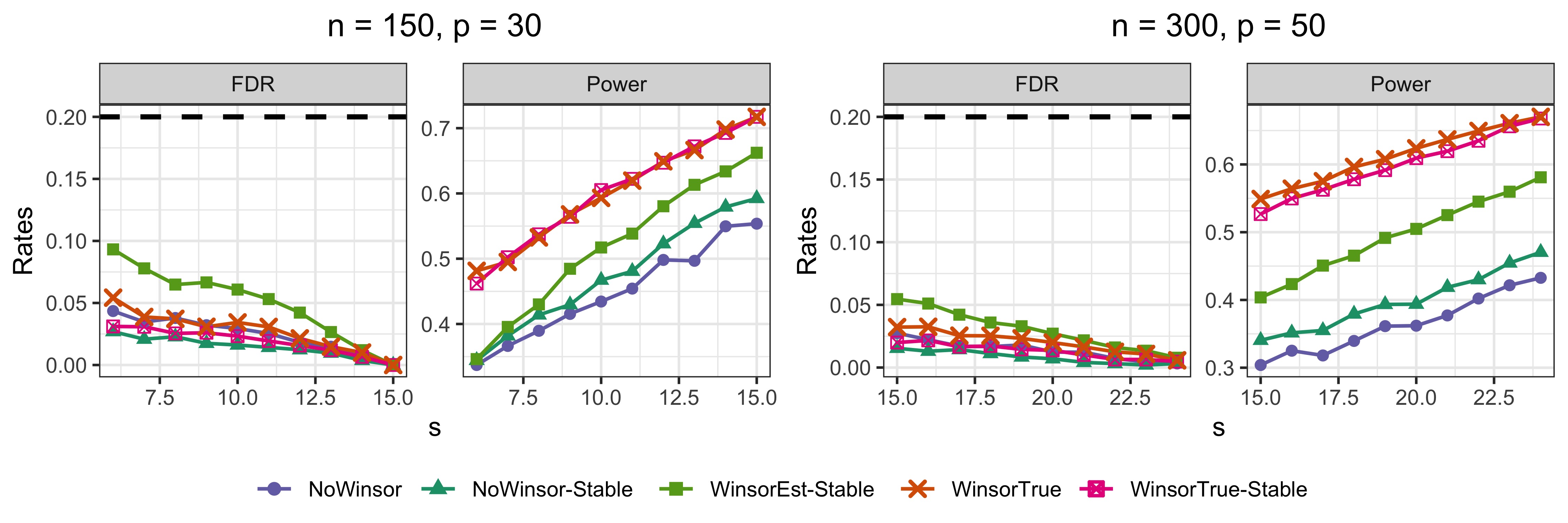}
    \end{subfigure}%
    \\
    \begin{subfigure}[t]{\textwidth}
        \centering
        \includegraphics[width=\textwidth, height=0.24\textheight]{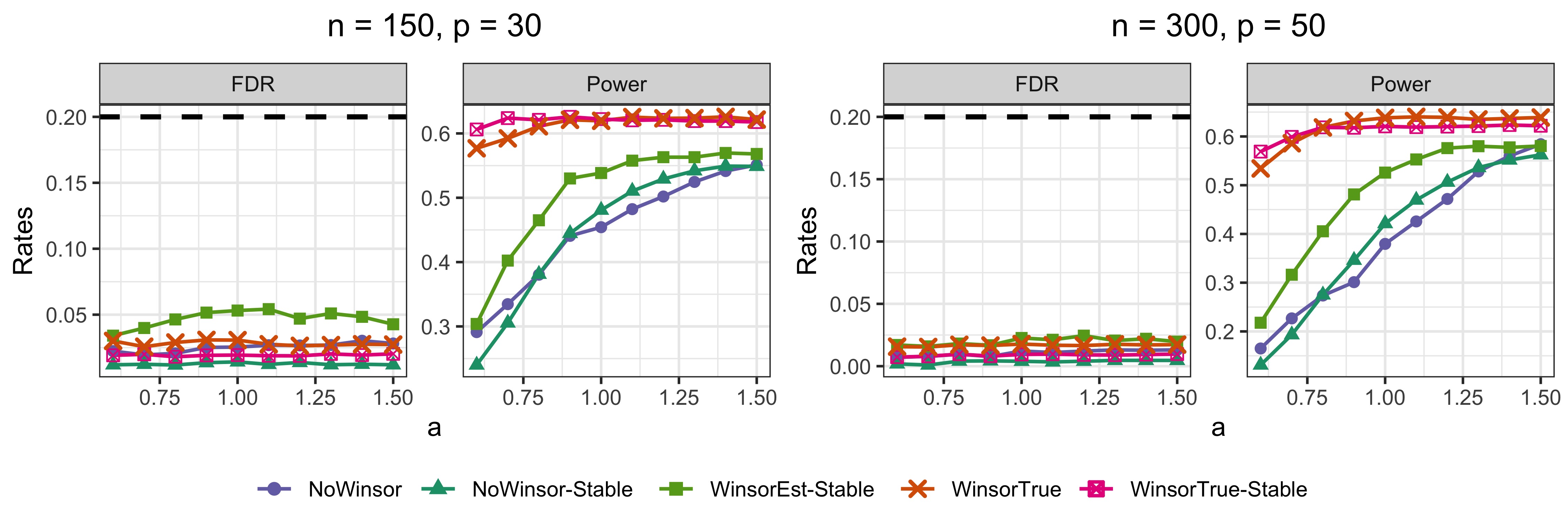}
    \end{subfigure}%
    \\
    \begin{subfigure}[t]{\textwidth}
        \centering
        \includegraphics[width=\textwidth, height=0.24\textheight]{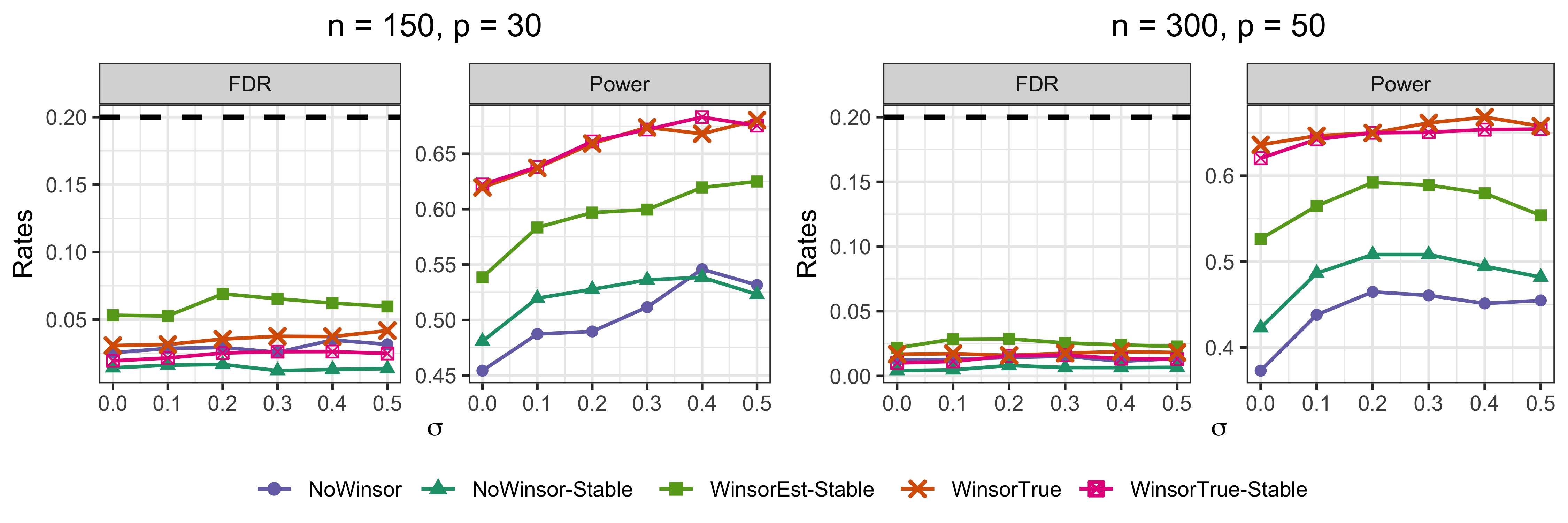}
    \end{subfigure}%
    \caption{Results obtained by varying $s$, $a$, $\sigma$ for $n \in \{150, 300\}$, and $p \in \{30, 50\}$. For all simulations, we fix $K = 0.5p$. }
    
    \label{fig4:ex3.2}
\end{figure}

Referring to the results for NoWinsor v.s. NoWinsor-Stable and WinsorTrue v.s. WinsorTrue-Stable, we see that the derandomization (Steps 4-5 from Algorithm \ref{c3:alg:RQR_knockoff_stable}) helps improve FDR and/or power in most settings. The results for Example 3 underscore the critical role of the winsorizing step. In most scenarios, the non-winsorized procedures failed to control the FDR at the target level. In contrast, the methods that employed winsorizing, both with true and estimated conditional quantiles, successfully maintained FDR control. Furthermore, the non-winsorized versions consistently demonstrated much lower selection power than the winsorized approaches. For Example 4, while all methods controlled the FDR, the non-winsorized procedures again exhibited substantially worse selection power. These two examples collectively highlight that our proposed winsorizing strategy is essential not only for achieving valid FDR control but also for preserving good statistical power.

\section{Application to LINE-1 Methylation Data}
\label{real:data:sec}

{

\subsection{Data Description}

This study uses data from The Cancer Genome Atlas (TCGA) head and neck squamous cell carcinoma (HNSC) cohort, available from the Genomic Data Commons (GDC). Two types of data are used:
(1) CpG methylation levels, which serve as covariates, and
(2) a derived measure of LINE-1 activity, based on the number of transposable element (TE) insertions detected from whole-genome sequencing.

{\bf Outcome (LINE-1 activity).}
The outcome variable quantifies LINE-1 activity using the total number of somatic TE insertions detected in each tumor genome. These counts were derived from whole-genome sequencing (WGS) data of the same TCGA HNSC samples, processed with the Sherlock-Lung WGS analysis pipeline and the TraFiC-mem pipeline \citep{rodriguez2020pan, zhang2021genomic, diaz2025mutagenic}. Because the dataset includes only tumor samples, most observations exhibit at least one insertion. However, the distribution of insertion counts is highly right-skewed: most tumors display few insertions, whereas a small subset shows substantially higher values.
 To stabilize variance and reduce skewness, we modeled the response variable as
$y=\log(\mbox {TE insertions} +1).$

{\bf Covariates (CpG methylation).}
CpG-level DNA methylation profiles were obtained from the Illumina HumanMethylation450 BeadChip (``450k array"), which assays over 450,000 CpG sites across the genome. We selected approximately 11,500 CpG probes located within promoter regions of retrotransposition-competent LINE-1 elements previously identified as active sources of somatic insertions \citep{rodriguez2020pan}. Each methylation level is expressed as a beta value between 0 (unmethylated) and 1 (fully methylated), representing the proportion of methylated alleles at that CpG site.  A two-step quality control procedure was applied to methylation data. 
(1) CpG sites with more than 10\% missing values across all samples were excluded. 
(2) For the remaining CpG sites, missing beta values were imputed using a local averaging approach, in which each missing value was replaced by the mean of its nearest neighboring CpG sites within the same sample. 

After processing the Head and Neck Squamous Cell Carcinoma (HNSC) cohort consisted of $n = 474$ tumor samples and $p = 11{,}484$ CpG sites. 
The samples were further divided into three groups as described in Section \ref{sec:algorithm}, Algorithm \ref{c3:alg:RQR_knockoff_stable} (e.g., $n_1 = 237$, $n_2 = 119$, $n_3 = 118$).

\subsection{Identifying CpG Sites Associated with Elevated LINE-1 Activity}

LINE-1 retrotransposon activity is a relatively recent focus in cancer genomics, but growing evidence indicates that aberrant activation of LINE-1 elements contributes to genomic instability and cancer progression \citep{tubio2014extensive, xiao2016line, rodriguez2020pan}. The activity of LINE-1 is often quantified through the number of TE insertions detected per sample, referred to as total TE insertions. This count variable ranges from zero to high values, where zero indicates no detectable activity. Consequently, total TE insertions are widely used as a surrogate measure for LINE-1 activity.

Previous studies \citep{iskow2010natural, tristan2020tumor} that included both normal and tumor samples or tumor-only samples typically used a binary characterization of LINE-1 activity, distinguishing whether any TE insertions were detected (total TE insertions = 0 vs. $>$0). This simplification was largely adopted because higher insertion counts are rare in most datasets. In contrast, our tumor-only data show that the distribution of total TE insertions among cancer patients is highly right-skewed, with values reaching up to 795. Limiting the analysis to a binary outcome, such as the presence or absence of insertions, would therefore overlook biologically relevant variation that may carry important information.

Given the limited literature on the distributional properties of total TE insertions, we employ a quantile-based framework to better characterize heterogeneity in LINE-1 activity. Specifically, we partition the outcome distribution of total TE insertions into three intervals corresponding to the 5th–35th percentile ($\Delta_1$), 35th–65th percentile ($\Delta_2$), and 65th–95th percentile ($\Delta_3$), representing low, moderate, and high LINE-1 activity, respectively. We then identify CpG sites whose methylation levels are associated with these distinct levels of LINE-1 activity. The tuning parameters for the StaRQR-K procedure were chosen to match those used in the simulation study ($R = 10$, $\eta = 0.5$, $q = 0.2$), and the screening threshold was set to $d = 2\lfloor n_1 / \log(n_1) \rfloor = 86$.

}

\subsection{Results for the Head and Neck Squamous Cell Carcinoma Cohort}

We applied the proposed StaRQR-K method to identify CpG sites whose methylation levels are associated with total TE insertions. For comparison, we also evaluated results from the initial screening step followed by the Derandomized Knockoff procedure, hereafter denoted as the \textit{screening + derandomized knockoff} approach. This benchmark also employs the non-crossing quantile LASSO to model regional effects but does not perform winsorization.

The selected CpG locations from both methods are summarized in Table \ref{c3:tab:sel_comparison}. Our StaRQR-K method identified numerous CpG locations with strong region-specific effects that were undetected by the derandomized knockoff procedure. Several CpG locations were associated with only one of $\Delta_1$, $\Delta_2$, or $\Delta_3$, whereas others were identified in two regions. 
\begin{table}[t]
\centering
\caption{Selected CpG locations of both methods for the HNSC cohort. It shows that several CpG locations are found in multiple regions and appear in both methods.}
\label{c3:tab:sel_comparison}
\begin{tabular}{c|p{6.25cm}|p{6.25cm}}
\hline\hline
 Region & Screening + Derandomized Knockoff & StaRQR-K \\ 
\hline
\aftercline
 Low ($\Delta_1$)  & {cg16543943} & cg03926968, cg07559427, cg13229363, cg14018786, cg15480200, {cg16543943} \beforecline 
\cline{1-3}\aftercline
                       Mid ($\Delta_2$)  & {cg16543943} & cg01106989, cg07483086, cg14191244, cg15841865, {cg16543943}, {cg16712481}, cg17660945, {cg22884082} \beforecline 
\cline{1-3}\aftercline
                       High ($\Delta_3$) & {cg03840912}, cg07483086, {cg16712481}, {cg23430295}, cg27660756 & {cg03840912}, cg04937596, cg07559427, cg11548446, cg14191244, {cg16712481}, {cg22884082}, {cg23430295} \beforecline 
\hline\hline
\end{tabular}
\end{table}

We computed the regional quantile prediction error (PE) on the testing data to assess model performance when predicting $y$ using CpG sites selected by StaRQR-K within each of $\Delta_1$, $\Delta_2$, and $\Delta_3$. 
Unlike the conventional mean squared prediction error, this metric is based on the quantile loss function, which evaluates how well the fitted conditional quantile functions approximate the observed responses across different parts of the outcome distribution. 
Intuitively, a smaller quantile-based PE indicates that the CpG sites selected by the method yield more accurate prediction of the conditional distribution of LINE-1 activity, not merely its mean level. 
This framework allows us to examine whether the selected CpG sites capture the heterogeneity of LINE-1 activity across lower, middle, and upper quantile regions. 
For comparison, we also calculated the corresponding PEs for the derandomized knockoff and the preliminary screening-only approaches.

Specifically, the regional quantile prediction error was defined as
\begin{align*}
\text{PE}(\Delta)
= \frac{\sum_{i=1}^n \mathbbm{1}\{\text{Subject } i \text{ in testing set}\} \sum_{l=1}^L \rho_{\tau_l}(y_i - \bx_i^T \hat{\beta}(\tau_l))}
{L \sum_{i=1}^n \mathbbm{1}\{\text{Subject } i \text{ in testing set}\}},
\end{align*}
where $\Delta = [\Delta_l, \Delta_u]$ with $\Delta_l < \tau_1 < \dots < \tau_L < \Delta_u$, and $\hat{\beta}(\cdot)$ denotes the regional nonparametric estimates from the training set based on B-spline bases \citep{schumaker2007spline}. The ratio of training to testing data was set to 0.8:0.2, and the process was repeated 100 times using randomly partitioned datasets. The results are shown in Figure \ref{c3:fig3:fig:prediction_error}.
\begin{figure}[t]
  \centering
  \includegraphics[width=0.9\textwidth]{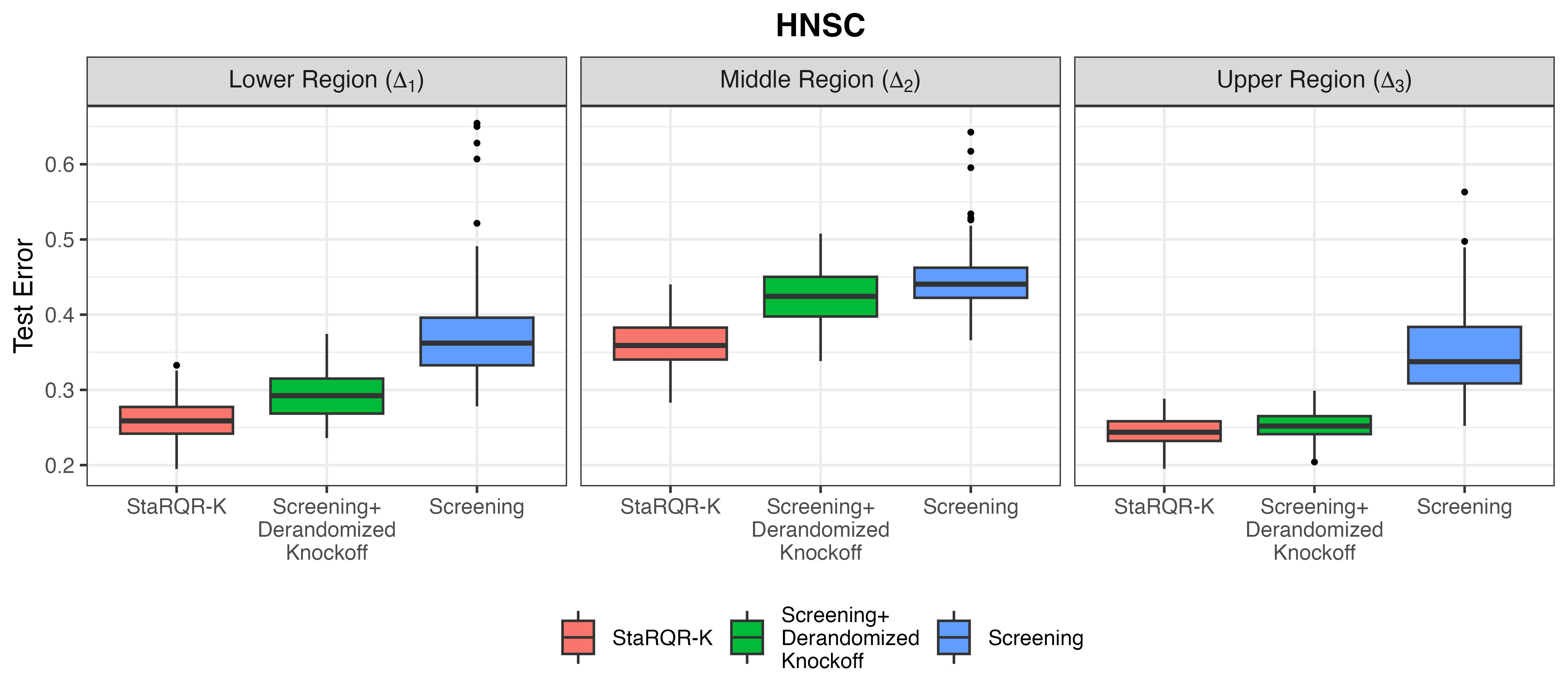}
  \caption{Out-of-sample prediction error boxplots from repeated train/test splits, comparing models built with variables selected by StaRQR-K, by the derandomized knockoff procedure, and by the preliminary screening step.}
  \label{c3:fig3:fig:prediction_error}
\end{figure}
The results indicate that the models constructed using CpG sites selected by StaRQR-K consistently achieve lower prediction error compared to those based on sites identified by the derandomized knockoff or screening-only approaches.

We further apply a regional quantile hypothesis test to assess the significance of each selected CpG site. The null and alternative hypotheses are formulated as follows:
\begin{align*}
    H_0:~\beta^*_j(\tau) = 0\text{ for all }\tau\in\Delta~~\text{ versus } ~~H_1:~\beta^*_j(\tau) \neq 0\text{ for some }\tau\in\Delta,\text{ and }j\in\{1,\dots,p\}.
\end{align*}
We conduct this test using the quantile rank score generated from a composite quantile regression estimator, which utilizes linear B-spline basis, by following the procedure described in Section 2 of \cite{PARK201713}. This procedure yields a test statistic that is asymptotically normal under $H_0$. We put the details of the procedure in the supplementary materials.


Table \ref{c3:tab:sel_cpg_test} presents the post-hoc test results and their scientific interpretation. 
Many of the selected CpG locations exhibited negative coefficients, consistent with the hypothesis that local demethylation accompanies increased LINE-1 activity. This observation aligns with the established role of CpG promoter demethylation in reactivating LINE-1 transcription, supporting the biological plausibility of our findings.

Several significant CpG locations were located near LINE-1 source elements previously identified as retrotransposition hotspots in PCAWG \citep{rodriguez2020pan}. For example, CpG cg07483086 on chromosome 2 lies near the promoter region of the 2q24.1 hotspot source element, linking this epigenetic mark to activation of a known driver of somatic LINE-1 insertions. These overlaps further reinforce the biological relevance of the associations detected by our framework. Post-hoc validation further confirmed that the CpG locations uniquely identified by the regional quantile regression framework are statistically significant. These validated CpG locations display negative coefficients, supporting a model in which CpG demethylation promotes LINE-1 mobilization. In practical terms, decreased methylation at these CpG locations corresponds to higher quantiles of LINE-1 activity, consistent with an activating role. Interestingly, the identified CpG locations exhibited quantile-specific effects, with some appearing only in the lower or upper quantile regions of the LINE-1 activity distribution. This pattern indicates that certain CpG sites influence low or high LINE-1 activity predominantly in tumors, rather than mean or uniformly across the entire distribution. A key biological insight from this analysis is that LINE-1 activation in tumors is not solely the result of diffuse global hypomethylation, which was the main result from the previous analyses, but can also be driven by targeted demethylation at hotspot source elements. This finding strengthens the mechanistic link between local epigenetic deregulation and the mobilization of specific LINE-1 source elements that dominate the retrotransposition landscape in cancer.

\begin{table}[ht]
\centering\footnotesize
\caption{Details of the selected CpG locations by quantile region. Distance denotes the distance to the closest LINE-1 genes. Sign denotes the sign of the estimated coefficients of the composite regional quantile regression model, and Test Stat and P-values are test statistics and p-values, respectively, defined in \cite{PARK201713}.}
\label{c3:tab:sel_cpg_test}
\begin{tabular}{c|cccccccc}
\hline\hline\aftercline
 Region & CpG ID & Chromosome & CpG position & Cytoband & Distance & \makecell{Sign \\ $~$} & \makecell{Test Stat \\ \citep{PARK201713}}  & \makecell{P-value \\ $~$}\\
\hline
\aftercline
  \multirow{6}{*}{Low ($\Delta_1$)}
   & cg03926968 & chr3  & 142440284 & q23   & 0  & (-) & 25.27 & $\leq$ 1e-100\\
  & cg07559427 & chr7  &  95467004 & q21.3 & 0  & (-) & 35.67 & $\leq$ 1e-100\\
  & cg13229363 & chr7  & 157961051 & q36.3 & 9  & (-) & 42.41 & $\leq$ 1e-100\\
  & cg14018786 & chr7  & 149159458 & q36.1 & 0  & (-) & 45.54 & $\leq$ 1e-100\\
  & cg15480200 & chr10 & 134630615 & q26.3 & 1  & (-) & 37.32 & $\leq$ 1e-100\\
  & cg16543943 & chr16 &  47216988 & q12.1 & 0  & (-) & 28.01 & $\leq$ 1e-100 \beforecline
  \cline{1-9}
  \aftercline
    \multirow{8}{*}{Mid ($\Delta_2$)}
   & cg01106989 & chr6  &  52858459 & p12.2 & 0  & (-) & 43.21 & $\leq$ 1e-100\\
  & cg07483086 & chr2  & 150444834 & q23.2 & 0  & (-) & 62.70 & $\leq$ 1e-100\\
  & cg14191244 & chr12 &  21547892 & p12.1 & 51 & (-) & 105.36 & $\leq$ 1e-100\\
  & cg15841865 & chr7  & 143173551 & q35   & 0  & (-) & 41.23 & $\leq$ 1e-100\\
  & cg16543943 & chr16 &  47216988 & q12.1 & 0  & (-) & 79.11 & $\leq$ 1e-100\\
  & cg16712481 & chr16 &  60556955 & q21   & 0  & (-) & 68.12 & $\leq$ 1e-100\\
  & cg17660945 & chr3  & 149511095 & q25.1 & 99 & (-) & 62.06 & $\leq$ 1e-100\\
  & cg22884082 & chr6  &  88039023 & q15   & 26 & (-) & 79.24 & $\leq$ 1e-100 \beforecline
  \cline{1-9}
    \aftercline
   \multirow{8}{*}{High ($\Delta_3$)}
   & cg03840912 & chr5  &  35494577 & p13.2 & 0  & (-) & 102.13 & $\leq$ 1e-100\\
  & cg04937596 & chr12 &  27236803 & p11.23& 27 & (-) & 62.06 & $\leq$ 1e-100\\
  & cg07559427 & chr7  &  95467004 & q21.3 & 0  & (-) & 98.27 & $\leq$ 1e-100\\
  & cg11548446 & chr2  &  28867073 & p23.2 & 86 & (-) & 47.55 & $\leq$ 1e-100\\
  & cg14191244 & chr12 &  21547892 & p12.1 & 51 & (-) & 91.00 & $\leq$ 1e-100\\
  & cg16712481 & chr16 &  60556955 & q21   & 0  & (-) & 69.70 & $\leq$ 1e-100\\
  & cg22884082 & chr6  &  88039023 & q15   & 26 & (-) & 67.35 & $\leq$ 1e-100\\
  & cg23430295 & chr7  & 140732738 & q34   & 0  & (-) & 43.38 & $\leq$ 1e-100 \\
\hline\hline
\end{tabular}
\end{table}

\section{Concluding Remarks}\label{sec:conclusion}

{ 

In this work, we introduced StaRQR-K, a statistical framework for variable selection in ultrahigh-dimensional regional quantile regression. The integration of regional quantile sure independence screening with a stabilized winsorizing-based knockoff filter effectively targets controlled false discovery at the prespecified level, which was not previously available for this class of models. The method showed strong performance in high-dimensional settings and produced biologically meaningful discoveries in the TCGA HNSC cohort. These results illustrate how quantile-specific variable selection can detect CpG associations with LINE1 activity that standard approaches miss. While our stabilization procedure significantly enhances selection power, the inherent reliance on data splitting for knockoff construction presents a limitation. We acknowledge that this splitting mechanism may attenuate finite sample guarantees and inevitably reduce statistical power in scenarios where the sample size is limited. Future work can extend the framework by exploring alternative screening strategies, improving knockoff construction under complex designs, and evaluating its robustness across diverse high-dimensional applications.



}

\section{Supplementary Materials} \label{sec:append}
\subsection{Proof of Theorem \ref{c3:thm:one}} \label{c3:proof:thm1}

Throughout the proof, we use $D_1,D_2,\ldots$ to denote constants that may depend on the fixed constants such as $\Delta_l, \Delta_u, K_x, k, \underline{f},\bar{f},\bar{f}', c_1, \kappa, \gamma$ from the main text. An explicit (though not optimal) dependence of $\{D_j, j=1,2,\ldots\}$ on the aforementioned constants can be tracked down. However, since it does not provide much more insight, we will often not present the explicit forms of $\{D_j, j=1,2,\ldots\}$, and this will greatly help streamline the proof. The constants $\{D_j,j=1,2,\ldots\}$ may vary from lines to lines.

\begin{proof}
Theorem \ref{c3:thm:one} is proved as follows:
\begin{align*}
\min_{j\in \mathcal{M}(\Delta)} \frac{1}{L}\sum_{\ell=1}^L\big(\mathbf{B}(\tau_{\ell})^T\hat{\bb}_j\big)^2 & \overset{(a)}{\geq} \min_{j\in \mathcal{M}(\Delta)}  \frac{1}{2L}\sum_{\ell=1}^L\big(\mathbf{B}(\tau_{\ell})^T\bb^*_j\big)^2- \max_{j\in \mathcal{M}(\Delta)} \frac{1}{L}\sum_{\ell=1}^L\big(\mathbf{B}(\tau_{\ell})^T(\hat{\bb}_j-\bb_j^*)\big)^2 \\
& \overset{(b)}{\geq} \frac{\kappa^2}{4}n^{-2\gamma}- \lambda_{\max}\Big(\frac{1}{L}\sum_{\ell=1}^L \bB(\tau_{\ell})\bB(\tau_{\ell})^T\Big) \cdot \max_{j\in \mathcal{M}(\Delta)} \|\hat{\bb}_j-\bb^*_j\|_2^2 \\
&  \overset{(c)}{\geq} \frac{\kappa^2}{4}n^{-2\gamma}- D_5^2c_n^2N^2n^{-1}\Big(\frac{D_2}{N}+\frac{D_3}{L}\Big) \\
&  \overset{(d)}{\geq} \frac{\kappa^2}{5}n^{-2\gamma}.
\end{align*}
Here, $(a)$ holds by the basic inequality $a^2\geq b^2/2-(a-b)^2$; $(b)$ is by Lemma \ref{main:result:2}; $(c)$ is due to Lemmas \ref{main:result1} and \ref{useful:1}; $(d)$ holds under Condition \ref{con:6} with the choice $c_n=\sqrt{2D_1^{-1}\log n}$.
\end{proof}

\begin{lemma}
\label{main:result1}
Assume Conditions \ref{con:1}, \ref{con:3}, \ref{con:4}, and $N=o(L), c_n Nn^{-1/2}=o(1), c_n\rightarrow \infty$. With probability at least $1-ne^{-D_1c_n^2}$, the following holds
\begin{align*}
\|\hat{\bb}_j-\bb_j^*\|_2 \leq D_5 c_n Nn^{-1/2}, \quad \forall j\in \mathcal{M}(\Delta).
\end{align*}
\end{lemma}

\begin{proof}
For notational simplicity, we will drop the subscript $j$ until near the end of the proof. Let $\btheta=(\ba, \bb)$ and introduce the following notations:
\begin{align*}
\ell_n(\btheta)&:=\frac{1}{nL}\sum_{\ell=1}^L\sum_{i=1}^n\rho_{\tau_{\ell}}\big(y_i-x_{i}\mathbf{B}(\tau_{\ell})^T\bb-\mathbf{B}(\tau_{\ell})^T\ba\big) \\
\ell(\btheta)&:=\frac{1}{L}\sum_{\ell=1}^L \mathbb{E}\rho_{\tau_{\ell}}\big(y-x \mathbf{B}(\tau_{\ell})^T\bb-\mathbf{B}(\tau_{\ell})^T\ba\big) \\
v_i(\btheta)&:=\frac{1}{L}\sum_{\ell=1}^L \Big[\rho_{\tau_{\ell}}\big(y_i-x_i \mathbf{B}(\tau_{\ell})^T\bb-\mathbf{B}(\tau_{\ell})^T\ba\big)-\rho_{\tau_{\ell}}\big(y_i-x_i \mathbf{B}(\tau_{\ell})^T\bb^*-\mathbf{B}(\tau_{\ell})^T\ba^*\big)\Big]
\end{align*}
A standard argument based on convexity of $\ell_n(\btheta)$, see \cite{hjort2011asymptotics} for example, shows that $\forall \delta>0$,
\begin{align}
\label{master:form}
&~\mathbb{P}\Big(\|\btheta-\btheta^*\|_2\geq \delta\Big) \nonumber \\
\leq &~ \mathbb{P}\Bigg(\inf_{\|\btheta-\btheta^*\|_2= \delta}\ell(\btheta)-\ell(\btheta^*)\leq  \sup_{\|\btheta-\btheta^*\|_2\leq \delta} \Big|\frac{1}{n}\sum_{i=1}^n(v_i(\btheta)-\mathbb{E}v_i(\btheta))\Big| \Bigg).
\end{align}
To utilize \eqref{master:form}, we first lower bound $\inf_{\|\btheta-\btheta^*\|_2= \delta}\ell(\btheta)-\ell(\btheta^*)$. Note that $\btheta^*\in \argmin_{\btheta}\ell(\btheta)$ according to \eqref{star:def}. This optimality together with the identity of Knight yields that for $\theta$ satisfying $\|\btheta-\btheta^*\|_2= \delta$,
\begin{align}
\label{lower:bound:ell}
 &~\ell(\btheta)-\ell(\btheta^*)  \nonumber \\
 =&~\frac{1}{L}\sum_{\ell=1}^L\mathbb{E}\int_0^{x\bB(\tau_{\ell})^T(\bb-\bb^*)+\bB(\tau_{\ell})^T(\ba-\ba^*)}\Big(F_{y|x}(x \mathbf{B}(\tau_{\ell})^T\bb^*+\mathbf{B}(\tau_{\ell})^T\ba^*+t)\nonumber \\
 &\hspace{2cm} -F_{y|x}(x \mathbf{B}(\tau_{\ell})^T\bb^*+\mathbf{B}(\tau_{\ell})^T\ba^*)\Big)dt \nonumber \\
 \geq &~ \frac{\underline{f}}{2L}\sum_{\ell=1}^L\mathbb{E}\big|x\bB(\tau_{\ell})^T(\bb^*-\bb)+\bB(\tau_{\ell})^T(\ba^*-\ba)\big|^2-\frac{\bar{f}'}{6L}\sum_{\ell=1}^L\mathbb{E}\big|x\bB(\tau_{\ell})^T(\bb^*-\bb)+\bB(\tau_{\ell})^T(\ba^*-\ba)\big|^3 \nonumber \\
 \geq &~ \Big( \frac{\underline{f}}{2}-\frac{\bar{f}'}{6}\delta (1+K_x) \Big)\cdot \frac{1}{L}\sum_{\ell=1}^L\Big((\bB(\tau_{\ell})^T(\bb^*-\bb))^2+(\bB(\tau_{\ell})^T(\ba^*-\ba))^2\Big) \nonumber \\
 \geq &~\Big( \frac{\underline{f}}{2}-\frac{\bar{f}'}{6}\delta (1+K_x) \Big)\delta^2 \lambda_{min}\Big(\frac{1}{L}\sum_{\ell=1}^L\bB(\tau_{\ell})\bB^T(\tau_{\ell})\Big).
 \end{align}
 Here, in the second inequality we have used Condition \ref{con:1} and $\|\bB(\tau_{\ell})\|_2\leq 1$ to employ Cauchy–Schwarz inequality. We now turn to upper bounding $G:=\sup_{\|\btheta-\btheta^*\|_2\leq \delta} \Big|\frac{1}{n}\sum_{i=1}^n(v_i(\btheta)-\mathbb{E}v_i(\btheta))\Big|$. Note that the $v_i(\theta)$'s are i.i.d. random functions, and 
 \begin{align*}
 \sup_{\|\btheta-\btheta^*\|_2\leq \delta}|v_i(\btheta)| &\leq \sup_{\|\btheta-\btheta^*\|_2\leq \delta} \frac{1}{L}\sum_{\ell=1}^L|x_i\bB(\tau_{\ell})^T(\bb-\bb^*)+\bB(\tau_{\ell})^T(\ba-\ba^*)| \leq (1+K_x)\delta.
 \end{align*}
 Hence, we can apply bounded difference concentration inequality to obtain
 \begin{align}
 \label{bound:diff:eq}
 \mathbb{P}\big(G\geq \mathbb{E}(G)+ t\big)\leq \exp\Big(\frac{-D_1nt^2}{\delta^2}\Big), \quad \forall t>0.
 \end{align}
 Moreover, we apply symmetrization and contraction to derive 
 \begin{align}
 \label{bound:exp}
 \mathbb{E}(G) &\leq 2\mathbb{E}\Big(\sup_{\|\btheta-\btheta^*\|_2\leq \delta} \Big|\frac{1}{n}\sum_{i=1}^n\epsilon_iv_i(\btheta)\Big|\Big) \nonumber \\
 &\leq \frac{2}{L}\sum_{\ell=1}^L\mathbb{E}\Big(\sup_{\|\btheta-\btheta^*\|_2\leq \delta} \Big|\frac{1}{n}\sum_{i=1}^n\epsilon_i\big[\rho_{\tau_{\ell}}\big(y_i-x_i \mathbf{B}(\tau_{\ell})^T\bb-\mathbf{B}(\tau_{\ell})^T\ba\big) \nonumber \\
 &\hspace{3cm} -\rho_{\tau_{\ell}}\big(y_i-x_i \mathbf{B}(\tau_{\ell})^T\bb^*-\mathbf{B}(\tau_{\ell})^T\ba^*\big)\big]\Big|\Big) \nonumber \\
 &\leq~\frac{4}{L}\sum_{\ell=1}^L \mathbb{E}\Big(\sup_{\|\btheta-\btheta^*\|_2\leq \delta} \Big|\frac{1}{n}\sum_{i=1}^n\epsilon_i\big(x_i\mathbf{B}(\tau_{\ell})^T(\bb-\bb^*)+\mathbf{B}(\tau_{\ell})^T(\ba-\ba^*)\big)\Big|\Big)  \nonumber \\
 &\leq~4 \delta  \mathbb{E}\sqrt{\big|\frac{1}{n}\sum_{i=1}^n\epsilon_ix_i\big|^2+\big|\frac{1}{n}\sum_{i=1}^n\epsilon_i\big|^2} \leq D_2 \delta n^{-1/2},
 \end{align}
 where $\epsilon_i's$ are independent symmetric Bernoulli variables, and in the fourth inequality we have used Cauchy–Schwarz inequality and the fact $\|\bB(\tau_{\ell})\|_2\leq 1$. Putting together \eqref{master:form}-\eqref{bound:exp} gives
 \begin{align*}
 &\mathbb{P}\Big(\|\hat{\bb}_j-\bb_j^*\|_2\geq \delta \Big) \\
 \leq & \mathbb{P}\Bigg((D_3-D_4\delta)\delta^2 \lambda_{min}\Big(\frac{1}{L}\sum_{\ell=1}^L\bB(\tau_{\ell})\bB^T(\tau_{\ell})\Big)\leq G {\rm~and~} G\leq D_2\delta n^{-1/2}+t \Bigg)+\exp\Big(\frac{-D_1nt^2}{\delta^2}\Big).
 \end{align*}
 Based on Lemma \ref{useful:1}, choosing $t=c_n\delta n^{-1/2}, \delta=D_5c_n N n^{-1/2}$ in the above inequality yields
 \[
 \mathbb{P}\Big(\|\hat{\bb}_j-\bb_j^*\|_2\geq D_5c_n N n^{-1/2} \Big) \leq \exp(-D_1 c_n^2),
 \]
 under the condition $N=o(L), c_n Nn^{-1/2}=o(1), c_n\rightarrow \infty$ and the constant $D_5$ is chosen large enough. Finally, employing a union bound over $j$ completes the proof.

\end{proof}

\begin{lemma}
\label{main:result:2}
Assume $N^{-d}=o(n^{-\gamma})$ and Conditions \ref{con:1}-\ref{con:5}.  It holds that 
\begin{align*}
\frac{1}{L}\sum_{\ell=1}^L(\bB(\tau_{\ell})^T\bb_j^*)^2 \geq \frac{\kappa^2}{2} n^{-2\gamma}, \quad \forall j\in \mathcal{M}(\Delta).
\end{align*}
\end{lemma}
\begin{proof}
Under Conditions \ref{con:2} and \ref{con:4}, it is known (Theorem 6.31 in \cite{schumaker2007spline}) that there exist $\ba_j^0,\bb_j^0\in \mathbb{R}^N$ and a constant $D_1>0$ such that 
\begin{align}
\label{approx:prop}
\sup_{\tau\in \Delta} |f_j(\tau)-\bB(\tau)^T\bb_j^0|\leq D_1 N^{-d}, \quad \sup_{\tau \in \Delta} |g_j(\tau)-\bB(\tau)^T\ba_j^0|\leq D_1 N^{-d}.
\end{align}
According to the definitions in \eqref{star:def} and \eqref{c3:pop:gf}, we have
\begin{align}
\label{chain:of:two}
\frac{1}{L}\sum_{\ell=1}^L\mathbb{E}\rho_{\tau_{\ell}}(y-x_jf_j(\tau_{\ell})-g_j(\tau_{\ell})) &\leq \frac{1}{L}\sum_{\ell=1}^L\mathbb{E}\rho_{\tau_{\ell}}(y-x_j\bB(\tau_{\ell})^T\bb_j^*-\bB(\tau_{\ell})^T\ba_j^*) \nonumber \\
&\leq \frac{1}{L}\sum_{\ell=1}^L\mathbb{E}\rho_{\tau_{\ell}}(y-x_j\bB(\tau_{\ell})^T\bb_j^0-\bB(\tau_{\ell})^T\ba_j^0).
\end{align}
Then using the identity of Knight like in \eqref{knight:first}, we obtain
\begin{align*}
0& \leq  \frac{1}{L}\sum_{\ell=1}^L\big[\mathbb{E}\rho_{\tau_{\ell}}(y-x_j\bB(\tau_{\ell})^T\bb_j^0-\bB(\tau_{\ell})^T\ba_j^0)-\mathbb{E}\rho_{\tau_{\ell}}(y-x_jf_j(\tau_{\ell})-g_j(\tau_{\ell}))\big] \\
&\leq  \frac{1}{L}\sum_{\ell=1}^L\Big(\frac{\bar{f}}{2}\mathbb{E}|x_j(\bB(\tau_{\ell})^T\bb_j^0-f_j(\tau_{\ell})) +\bB(\tau_{\ell})^T\ba_j^0-g_j(\tau_{\ell})|^2\\
&\hspace{2cm} +\frac{\bar{f}'}{6}\mathbb{E}|x_j(\bB(\tau_{\ell})^T\bb_j^0-f_j(\tau_{\ell}))+\bB(\tau_{\ell})^T\ba_j^0-g_j(\tau_{\ell})|^3\Big) \\
&\leq D_2 N^{-2d}  \quad {\rm when~}N{\rm~is~large},
\end{align*}
where the last inequality is due to \eqref{approx:prop} and Condition \ref{con:1}. The above combined with \eqref{chain:of:two} yields
\begin{align}
\label{upper:use:exp}
\frac{1}{L}\sum_{\ell=1}^L\big[\mathbb{E}\rho_{\tau_{\ell}}(y-x_j\bB(\tau_{\ell})^T\bb_j^*-\bB(\tau_{\ell})^T\ba_j^*)-\mathbb{E}\rho_{\tau_{\ell}}(y-x_jf_j(\tau_{\ell})-g_j(\tau_{\ell}))\big] \leq D_2 N^{-2d}.
\end{align}
Lemma \ref{useful:2} with $\tau$ uniformly distributed over $\{\tau_{\ell}\}_{\ell=1}^L$, together with \eqref{upper:use:exp}, shows that when $N$ is large, 
\begin{align*}
&~\frac{\underline{f}}{4}\Big(\frac{1}{L}\sum_{\ell=1}^L|\bB(\tau_{\ell})^T\bb_j^*-f_j(\tau_{\ell})|\Big)^2 \\
\leq &~\frac{1}{L}\sum_{\ell=1}^L\big[\mathbb{E}\rho_{\tau_{\ell}}(y-x_j\bB(\tau_{\ell})^T\bb_j^*-\bB(\tau_{\ell})^T\ba_j^*)-\mathbb{E}\rho_{\tau_{\ell}}(y-x_jf_j(\tau_{\ell})-g_j(\tau_{\ell}))\big] \\
\leq & D_2 N^{-2d},
\end{align*}
which implies 
\[
\frac{1}{L}\sum_{\ell=1}^L|\bB(\tau_{\ell})^T\bb_j^*-f_j(\tau_{\ell})|\leq D_3 N^{-d}. 
\]
We thus conclude
\begin{align}
\label{last2:result}
\sqrt{\frac{1}{L}\sum_{\ell=1}^L(\bB(\tau_{\ell})^T\bb_j^*)^2}&\geq \frac{1}{L}\sum_{\ell=1}^L|\bB(\tau_{\ell})^T\bb_j^*| \nonumber \\
&\geq \frac{1}{L}\sum_{\ell=1}^L|f_j(\tau_{\ell})|-\frac{1}{L}\sum_{\ell=1}^L|\bB(\tau_{\ell})^T\bb_j^*-f_j(\tau_{\ell})| \geq \kappa n^{-\gamma}-D_3N^{-d},
\end{align}
where the first inequality is by Cauchy–Schwarz inequality, the second is by triangle inequality, and the last one is from Condition \ref{con:5}.
\end{proof}

\begin{lemma}
\label{useful:2}
Under Conditions \ref{con:1} and \ref{con:3}, it holds that $\forall j\in \mathcal{M}(\Delta)$,
\begin{align*}
&~\mathbb{E}\rho_{\tau}(y-x_jf(\tau)-g(\tau))-\mathbb{E}\rho_{\tau}(y-x_j f_j(\tau)-g_j(\tau)) \\
\geq &~
\begin{cases}
\frac{\underline{f}}{4}(\mathbb{E}|f(\tau)-f_j(\tau)|)^2 & \text{if~} \mathbb{E}|f(\tau)-f_j(\tau)| \leq \frac{q}{2},  \\
\frac{\underline{f}}{4} (q\mathbb{E}|f(\tau)-f_j(\tau)|-\frac{q^2}{4})& \text{if~} \mathbb{E}|f(\tau)-f_j(\tau)| > \frac{q}{2},
\end{cases}
\end{align*}
where $q=\frac{3\underline{f}}{2\bar{f}'(1+K_x)}$, and  the expectation $\mathbb{E}(\cdot)$ is allowed to be taken additionally with respect to $\tau$ that is independent from $(y, x_j)$.
\end{lemma}

\begin{proof}
The proof is motivated by the proof of Lemma 4 in \cite{belloni2011l1}. We use $\mathbb{E}_{-\tau}$ to denote the expectation taken only with respect to $(y,x_j)$. Using the identity of Knight, we first obtain
\begin{align}
\label{knight:first}
&~\mathbb{E}_{-\tau}\rho_{\tau}(y-x_jf(\tau)-g(\tau))-\mathbb{E}_{-\tau}\rho_{\tau}(y-x_j f_j(\tau)-g_j(\tau)) \nonumber \\
=&~ \mathbb{E}_{-\tau}\Big(\big[x_j(f_j(\tau)-f(\tau))+g_j(\tau)-g(\tau)\big]\cdot \big[\tau-\mathbbm{1}(y\leq x_jf_j(\tau)+g_j(\tau))\big]\Big)  \nonumber \\
&~+\mathbb{E}_{-\tau}\int_0^{x_j(-f_j(\tau)+f(\tau))-g_j(\tau)+g(\tau)}\big(F_{y|x_j}(x_jf_j(\tau)+g_j(\tau)+t)-F_{y|x_j}(x_jf_j(\tau)+g_j(\tau))\big)dt  \nonumber \\
 \geq&~ \frac{\underline{f}}{2}\mathbb{E}_{-\tau}|x_j(f(\tau)-f_j(\tau))+g(\tau)-g_j(\tau)|^2-\frac{\bar{f}'}{6}\mathbb{E}_{-\tau}|x_j(f(\tau)-f_j(\tau))+g(\tau)-g_j(\tau)|^3, \nonumber \\
 \geq&~\frac{\underline{f}}{2}\big[(f(\tau)-f_j(\tau))^2+(g(\tau)-g_j(\tau))^2\big]\cdot \Big[1-\frac{\bar{f}'(1+K_x)}{3\underline{f}}\sqrt{(f(\tau)-f_j(\tau))^2+(g(\tau)-g_j(\tau))^2}\Big],
\end{align} 
where in the first inequality we have used Taylor expansion and Condition \ref{con:3} to bound the second term, and the first term equals zero due to the definition of $f_j(\tau),g_j(\tau)$ in \eqref{c3:pop:gf}; the last inequality is due to Condition \ref{con:1} and the fact that $x_j$ is standardized. When $\mathbb{E}_{-\tau}|x_j(f(\tau)-f_j(\tau))+g(\tau)-g_j(\tau)|^2\leq q^2$, \eqref{knight:first} implies
\begin{align}
\label{result:first}
&~\mathbb{E}_{-\tau}\rho_{\tau}(y-x_jf(\tau)-g(\tau))-\mathbb{E}_{-\tau}\rho_{\tau}(y-x_j f_j(\tau)-g_j(\tau)) \nonumber \\
\geq &~\frac{\underline{f}}{4} \cdot \mathbb{E}_{-\tau}|x_j(f(\tau)-f_j(\tau))+g(\tau)-g_j(\tau)|^2.
\end{align}
When $\mathbb{E}_{-\tau}|x_j(f(\tau)-f_j(\tau))+g(\tau)-g_j(\tau)|^2>q^2$, define 
\[
\bar{f}(\tau)=t f(\tau)+(1-t)f_j(\tau), \quad \bar{g}(\tau)=tg(\tau)+(1-t)g_j(\tau), 
\]
with $t=\frac{q}{\sqrt{\mathbb{E}_{-\tau}|x_j(f(\tau)-f_j(\tau))+g(\tau)-g_j(\tau)|^2}} \in (0,1)$. Using the convexity of $\rho_{\tau}(\cdot)$ we have
\begin{align}
\label{result:second}
&~\mathbb{E}_{-\tau}\rho_{\tau}(y-x_jf(\tau)-g(\tau))-\mathbb{E}_{-\tau}\rho_{\tau}(y-x_j f_j(\tau)-g_j(\tau))\nonumber \\
\geq &~\frac{1}{t} \cdot \big(\mathbb{E}_{-\tau}\rho_{\tau}(y-x_j\bar{f}(\tau)-\bar{g}(\tau))-\mathbb{E}_{-\tau}\rho_{\tau}(y-x_j f_j(\tau)-g_j(\tau))\big) \nonumber \\
\geq &~  \frac{\underline{f}q}{4}\sqrt{\mathbb{E}_{-\tau}|x_j(f(\tau)-f_j(\tau))+g(\tau)-g_j(\tau)|^2},
\end{align}
where the second inequality holds by invoking \eqref{result:first} since  $\mathbb{E}_{-\tau}|x_j(\bar{f}(\tau)-f_j(\tau))+\bar{g}(\tau)-g_j(\tau)|^2= q^2$. Denote $h_1(t):=\frac{\underline{f}}{4}\min(t^2,tq)$ and
\begin{align*}
\label{h2:def}
h_2(t):=
\begin{cases}
\frac{\underline{f}}{4}t^2 & \text{if~} t\leq \frac{q}{2} \\
\frac{\underline{f}}{4}(qt-\frac{q^2}{4}) & \text{if~}t>\frac{q}{2}
\end{cases}
\end{align*}
It is straightforward to verify that \eqref{result:first} and \eqref{result:second} together imply
\begin{align*}
&~\mathbb{E}_{-\tau}\rho_{\tau}(y-x_jf(\tau)-g(\tau))-\mathbb{E}_{-\tau}\rho_{\tau}(y-x_j f_j(\tau)-g_j(\tau)) \\
\geq &~h_1(|f(\tau)-f_j(\tau)|) \geq h_2(|f(\tau)-f_j(\tau)|).
\end{align*}
Moreover, $h_2(t)$ is convex over $t\in [0,\infty)$. Hence, we take expectation with respect to $\tau$ for the above inequality and apply Jensen's inequality to conclude
\begin{align*}
\mathbb{E} \rho_{\tau}(y-x_jf(\tau)-g(\tau))-\mathbb{E}\rho_{\tau}(y-x_j f_j(\tau)-g_j(\tau)) \geq h_2(\mathbb{E}|f(\tau)-f_j(\tau)|).
\end{align*}
\end{proof}

\begin{lemma}
\label{useful:1}
Under Condition \ref{con:4}, the B-spline basis vector $\mathbf{B}(\tau)=(B_1(\tau),\ldots, B_N(\tau))$ satisfies
\begin{align*}
\frac{D_1}{N}-\frac{D_3}{L}\leq\lambda_{\min}\Big(\frac{1}{L}\sum_{\ell=1}^L\mathbf{B}(\tau_{\ell})\mathbf{B}(\tau_{\ell})^T\Big )\leq \lambda_{\max}\Big(\frac{1}{L}\sum_{\ell=1}^L\mathbf{B}(\tau_{\ell})\mathbf{B}(\tau_{\ell})^T\Big)\leq \frac{D_2}{N}+\frac{D_3}{L},
\end{align*}
where $D_1,D_2, D_3>0$ are constants depending on $k, c_1,\Delta_l,\Delta_u$. 

\end{lemma}
\begin{proof}
The proof is a direct application of Lemma 6.1 in \cite{shen1998local}. Adopting the notations therein, the distributions $Q_n(x)$ and $Q(x)$ in the current setting are
\begin{align*}
Q_n(x)=\frac{1}{L}\sum_{\ell=1}^L \mathbbm{1}(\tau_{\ell}\leq x), \quad Q(x)=\frac{x-\Delta_l}{\Delta_u-\Delta_l}, \quad {\rm~for~} x\in [\Delta_l,\Delta_u].
\end{align*}
Since $\{\tau_{\ell}\}_{\ell=1}^L$ are uniformly spaced over $[\Delta_l,\Delta_u]$, it is straightforward to verify that
\[
\sup_{x\in [\Delta_l,\Delta_u]}|Q_n(x)-Q(x)|\leq \frac{2}{L}.
\]
The rest of the proof directly follows the proof of Lemma 6.1 in \cite{shen1998local}.
\end{proof}

\subsection{Proof of Proposition 2}
We condition on $\bx$ throughout the proof, and drop the conditioning notation for simplicity. Note that the random variable $\tilde{y}$ is defined over $[Q_{y|\textbf{x}} (\Delta_l), Q_{y|\textbf{x}} (\Delta_u)]$, and $\mathbb{P}(\tilde{y}\leq t)=\mathbb{P}(y\leq t)$ for $t\in [Q_{y|\textbf{x}} (\Delta_l), Q_{y|\textbf{x}} (\Delta_u))$ and $\mathbb{P}(\tilde{y}\leq t)=1$ for $t=Q_{y|\textbf{x}} (\Delta_u)$. As a result, it is direct to verify that $Q_{\tilde{y}|\bx}(\tau)=Q_{y|\bx}(\Delta_l)$ for $\tau \leq \Delta_l$ and $Q_{\tilde{y}|\bx}(\tau)=Q_{y|\bx}(\Delta_u)$ for $\tau \geq \Delta_u$. Now for $\tau \in (\Delta_l, \Delta_u)$, if $Q_{y|\bx}(\tau)<Q_{y|\bx}(\Delta_u)$, then $Q_{\tilde{y}|\bx}(\tau)=Q_{y|\bx}(\tau)$ since $\mathbb{P}(\tilde{y}\leq t)=\mathbb{P}(y\leq t)$ near $t=Q_{y|\bx}(\tau)$. If $Q_{y|\bx}(\tau)=Q_{y|\bx}(\Delta_u)$, then we conclude $Q_{\tilde{y}|\bx}(\tau)=Q_{y|\bx}(\Delta_u)$. Otherwise, there exists $t^*<Q_{y|\bx}(\Delta_u)$ such that $\mathbb{P}(y\leq t^*)=\mathbb{P}(\tilde{y}\leq t^*)\geq \tau$ which contradicts with $Q_{y|\bx}(\tau)=Q_{y|\bx}(\Delta_u)>t^*$. The last statement about conditional independence follows because there is a one-to-one mapping between quantile and distribution functions.

\subsection{Simulation Result of Screening Performance for $\mathbf{\textit{p}=5,000}$}
\begingroup
\renewcommand{\arraystretch}{1.5}
\begin{table}[!htbp] 
\centering 
\caption{Minimum model size (MMS) result with $p=5,000$ and $n=300$. The QPCS method can only track at most $n=300$ top variables, so if the MMS exceeded the sample size, the result is reported as `300+'.} 
\label{c3:tab:example2_p5000} 
\begin{tabular}{@{\extracolsep{0.1pt}} cc| cccc | cccc} 
\hline 
\hline 
& &\multicolumn{4}{c|}{$\sigma = 0.5$}  &\multicolumn{4}{c}{$\sigma = 0.8$}\\ 
& &5\% & 35\% & 65\% & 95\% &5\% & 35\% & 65\% & 95\% \\
\hline 
rqSIS & Example 1 & 4 & 4 & 4 & 4 & 4 & 4 & 5 & 7 \\ 
$(\Delta)$ & Example 2 & 4 & 4 & 4 & 4 & 4 & 4 & 4 & 6 \\ \hline
DC-SIS & Example 1 & 51.4 & 445.4 & 1178.4 & 3506 & 13 & 122.9 & 731.3 & 3113.9 \\ 
 & Example 2 & 82 & 454.3 & 994.45 & 2602 & 20.95 & 149.3 & 429.4 & 2123.5 \\ \hline
PC-Screen & Example 1 & 98 & 642 & 1490.6 & 3585 & 14.7 & 153.9 & 826.2 & 3164.9 \\ 
 & Example 2 & 167.8 & 717.6 & 1472.75 & 3005.7 & 25.95 & 263 & 768.8 & 2599.1 \\ \hline 
qaSIS & Example 1 & 83.2 & 1022 & 2555.8 & 4569.4 & 124.8 & 1399.8 & 2972.4 & 4580.3 \\  
 $(\tau=0.25)$ & Example 2 & 431.7 & 2049.65 & 3459.65 & 4794.6 & 449.95 & 2422.3 & 3728.75 & 4826.4 \\ \hline
qaSIS  & Example 1 & 148.4 & 1539.2 & 2946.8 & 4677 & 33.1 & 541 & 1927.2 & 4162.4 \\ 
 $(\tau=0.5)$ & Example 2 & 587.9 & 2300.2 & 3633.1 & 4784.6 & 123.95 & 1354.65 & 2995.35 & 4686.3 \\ \hline
qaSIS  & Example 1 & 55.2 & 870.2 & 2715 & 4675.6 & 16 & 229.9 & 1182 & 3693.9 \\ 
$(\tau=0.75)$ & Example 2 & 335.75 & 1945.25 & 3421.65 & 4724.05 & 32 & 353.95 & 1401.4 & 4265.2 \\ \hline
QPCS  & Example 1 & 4 & 4 & 4 & 300+ & 4 & 4 & 4 & 300+ \\ 
$(\tau=0.25)$ & Example 2 & 50 & 300+ & 300+ & 300+ & 68.95 & 300+ & 300+ & 300+ \\ \hline
QPCS  & Example 1 & 220.4 & 300+ & 300+ & 300+ & 219.3 & 300+ & 300+ & 300+ \\ 
$(\tau=0.5)$ & Example 2 & 300+ & 300+ & 300+ & 300+ & 300+ & 300+ & 300+ & 300+ \\ \hline
QPCS  & Example 1 & 4 & 4 & 4 & 300+ & 4 & 4 & 4 & 300+ \\ 
$(\tau=0.75)$ & Example 2 & 67 & 300+ & 300+ & 300+ & 59.85 & 300+ & 300+ & 300+ \\ 
\hline \hline \\[-1.8ex] 
\end{tabular} 
\end{table} 
\endgroup

\subsection{Some Details for Regional Quantile Hypothesis Test in Section 4}
This procedure is similar to the procedure described in Section 2 of \cite{PARK201713}. For the details of the method, please see the original paper. Here, we are typically interested in testing the following hypothesis:
\begin{align*}
    H_0:~\beta^*_j(\tau) = 0\text{ for all }\tau\in\Delta~~\text{ versus } ~~H_1:~\beta^*_j(\tau) \neq 0\text{ for some }\tau\in\Delta,\text{ and }j\in\{1,\dots,p\}.
\end{align*}
When the null hypothesis is true, we consider the marginal composite quantile regression estimator as follows:
\begin{align}
    \hat{\ba} = \argmin_{\ba \in \mathbb{R}^N} \frac{1}{nL}\sum_{\ell=1}^L\sum_{i=1}^n\rho_{\tau_{\ell}}\big(y_i-\mathbf{B}(\tau_{\ell})^T\ba\big). \label{eq:supp_test}
\end{align}
Then, we define the conditional density matrix and projection matrix as follows:
\begin{align*}
    \bH^{(l)} &= \text{diag}\Big( f_1^{(\tau_l)}(0),\dots,f_n^{(\tau_l)}(0) \Big) \\
    \bP^{(l)} &= \Big((\bh^{(l)})^T\bh^{(l)}\Big)^{-1} \bh^{(l)} (\bh^{(l)})^T
\end{align*}
where $\bh^{(l)} = \bH^{(l)}\mathbf{1}_n$, $f$ are the conditional density functions, and $\mathbf{1}_n \in \mathbb{R}^n$ is the vector of ones. $f$ can be estimated by the methods described in the paper, and $\bH^{(l)}$ can be the identity matrix for the homoskedastic errors, as a special case. We also define $\bx_j = (x_{1j},\dots,x_{nj})^T$, and the projected covariates as $\tilde{\bx}_j^{(l)} = (I_n - \bP^{(k)})\bx_j$.
Then, the quantile rank score test rejects the null hypothesis if $S_{n}^T Q_{n}^{-1}S_n$ is greater than its critical value with
\begin{align*}
    S_n &= \frac{1}{L\sqrt{n}}\sum_{l=1}^L (\tilde{\bx}_j^{(l)})^T\psi_{\tau_l}(\by - \mathbf{1}_n \bB(\tau_l)^T \hat{\ba})\\
    Q_n &= \frac{1}{L^2n}\sum_{l=1}^L\sum_{l'=1}^L(\tau_l \wedge\tau_{l'} - \tau_l\tau_{l'}) (\tilde{\bx}_j^{(l)})^T \tilde{\bx}_j^{(l')}
\end{align*}
where $\psi_{\tau}(\bx) = (\psi_{\tau}(x_1),\dots,\psi_{\tau}(x_n))^T$ for any vector $\bx \in \mathbb{R}^n$ and $\psi_{\tau}(x) = \tau - \mathbbm{1}(x \leq 0)$ for any scalar $x$. The score follows an asymptotical normal distribution, so we reject it based on the critical value from the normal distribution.

\bibliography{mc}
\bibliographystyle{ims}

\end{document}